\documentclass[preprint,12pt]{elsarticle}

\usepackage{enumitem}
\usepackage{amssymb}
\usepackage{amsthm}
\usepackage{amsmath}
\usepackage{color}
\usepackage{siunitx}

\DeclareSIUnit{\mile}{mile}

\usepackage{hyperref}
\usepackage{graphicx}
\usepackage{subcaption}
\usepackage{tabularx}
\usepackage{booktabs}
\usepackage{float}

\newtheorem{lemma}{Lemma}[section]
\newtheorem{proposition}{Proposition}[section]

\newtheorem{definition}{Definition}

\journal{Physica A: Statistical Mechanics and its Applications}

\begin{document}

\begin{frontmatter}



\title{Hysteresis Behind A Freeway Bottleneck With Location-Dependent Capacity}


\author[inst1]{Alexander Hammerl\corref{cor1}}
\cortext[cor1]{Corresponding author. Email: asiha@dtu.dk}

\affiliation[inst1]{organization={Technical University of Denmark, Department of Technology, Management and Economics},
            addressline={Bygningstorvet 358}, 
            city={Kgs. Lyngby},
            postcode={2800}, 
            country={Denmark}}

\author[inst1]{Ravi Seshadri}
\author[inst1]{Thomas Kjær Rasmussen}
\author[inst1]{Otto Anker Nielsen}

\begin{abstract}
Macroscopic fundamental diagrams (MFDs) and related network traffic dynamics models have received both theoretical support and empirical validation with the emergence of new data collection technologies. However, the existence of well-defined MFD curves can only be expected for traffic networks with specific topologies and is subject to various disturbances, most importantly hysteresis phenomena. This study aims to improve the understanding of hysteresis in Macroscopic Fundamental Diagrams and Network Exit Functions (NEFs) during rush hour conditions. We apply the LWR theory to a highway corridor featuring a location-dependent downstream bottleneck to identify a figure-eight hysteresis pattern, clockwise on the top and counter-clockwise on the bottom. Our empirical observations confirm the occurrence of counter-clockwise loops in real conditions, an effect which we can attribute to demand asymmetries through theoretical analysis. The paper discusses the impact of the road topology and demand patterns on the formation and intensity of hysteresis loops analytically. To substantiate these findings, we analyze empirical MFD data from two bottlenecks and present statistical evidence that, under otherwise identical conditions, a continuous bottleneck causes less hysteresis than a discontinuous one. We conduct numerical experiments using the Cell Transmission Model (CTM) to show that even a slight reduction in the capacity of the homogeneous section can significantly decrease MFD hysteresis while maintaining outflow at the corridor's downstream end. These reductions can be achieved with minimal intervention through standard traffic control measures, such as dynamic speed limits or ramp metering.
\end{abstract}

\begin{highlights}
\item Show figure-eight hysteresis in link MFD using LWR theory under rush-hour conditions
\item Analyze impact of demand patterns and road layouts on MFD hysteresis
\item Prove triangular FD with high jam density maximizes hysteresis
\item Empirical validation of the influence of bottleneck geometry on the extent of hysteresis
\item Numerical simulations show how bottleneck configuration leads to less hysteresis and smother traffic
\end{highlights}


\end{frontmatter}


\section{Introduction}
\label{sec:introduction}
\cite{God69} is considered the first author to postulate a strictly functional, uni-modal relationship between average flow and density in urban traffic networks. After traffic flow theory \cite{DagGer08} and empirical data \cite{GerDag08} have more recently confirmed the existence of such relationships, their estimation and analysis have gained considerable popularity. Such relationships are commonly referred to as Macroscopic Fundamental Diagrams (MFDs). A similar relationship, known as the Network Exit Function (NEF), has also been demonstrated between trip completion rate and average density \cite{Dag07}. 

In a recent review paper, \cite{johal21} categorize the requirements for the existence of MFDs as follows: demand homogeneity, road homogeneity, and control homogeneity. In particular, assuming demand and road homogeneity is rarely feasible in practical transport network modeling without significant restrictions.  \cite{builad09} show that the inhomogeneity of density within a network during rush hour can influence the relationship between traffic flow and density in such a way that clockwise hysteresis loops may form. Next, \cite{masal10} and \cite{yial10} and later \cite{mahal13}, \cite{leqal15} and \cite{Yuan2023} confirmed that congestion distribution influences the emergence of MFD hysteresis loops. The significance of congestion distribution for the MFD's shape and scatter has sparked interest in how this can be appropriately addressed (e.g. \cite{yiger12}). \cite{Knoop2015} pursue a generalized approach, where network flow is modeled as a multivariate function of both average density and vehicle distribution. \cite{Xu2023} focus on explaining hysteresis loops in NEFs under the assumption of a single discontinuous bottleneck, both on a single corridor and for idealized entire networks. They identify the potential occurrence of counter-clockwise hysteresis loops and figure-eight hysteresis loops, but do not consider the influence of congestion effects on these phenomena. 

Other sources attribute the occurrence of hysteresis loops to the topology of the respective road networks. As early as \cite{builad09} and \cite{yial10}, it was indicated that hysteresis effects are to be expected in freeway networks. \cite{gersun11}, \cite{sabmah12}, and \cite{sabmah13} provide a detailed discussion about freeway MFDs. Daganzo (\cite{Daganzo2011}) observes that a corridor's MFD "will exhibit clockwise hysteresis during any demand-driven episode of bottleneck queuing". Bottlenecks in traffic flow can stem from a broad variety of causes, and empirical observations for many of them suggest a continuous transition between regular traffic conditions and the reduced capacity at the bottleneck, including merges \cite{Daamen2010}, diverges \cite{Munoz2002} and lane drops \cite{Bertini2005}. Therefore, in this article, we study hysteresis phenomena in a freeway corridor with location-dependent capacity, where we model the capacity of the bottleneck(s) as a non-increasing function in space, $q_{bn}(x)$, rather than adopting the more common approach of representing the bottleneck as a discontinuous drop in capacity at a single point. We build on Daganzo's observation of clockwise hysteresis in a corridor MFD and contribute to the literature in several respects. 

First, we apply the LWR theory to a highway corridor featuring a location-dependent downstream bottleneck to identify a figure-eight hysteresis pattern, clockwise on the top and counter-clockwise on the bottom.
A higher average flow at a given average density in the offset of congestion is possible outside the interval of active queuing and in special cases, even when an active queue exists at both time points. The resulting shape of the MFD is a figure eight, with the upper loop moving clockwise and the lower loop moving counter-clockwise. We also discuss the roles of the shape of the fundamental diagram and the pattern of demand on hysteresis and why this general pattern is rare in practical scenarios, where a single clockwise loop is more typical. Second, we analytically characterize the impact of changes in peak demand on the area of the hysteresis loop for the case of a single discontinuous bottleneck and show that the area under the hysteresis curve is maximum for a triangular fundamental diagram with arbitrarily high jam density. Third, we use empirical data from two study sites (Interstate 880 North and Highway 41 in California) and simulations using the Cell Transmission Model (CTM) to validate 1) the theoretical analysis on the shapes of hysteresis loops, and 2) the hypothesis that a bottleneck with gradually decreasing geometry generates less hysteresis under otherwise identical conditions. The latter has useful implications for traffic control, which we demonstrate through extensive experiments using the CTM.     


Although the LWR theory in its general formulation is more expressive than conventional queuing systems,  queueing models provide a widely accepted framework for modeling peak hour traffic, see for example \cite{Vickrey1969}, \cite{Ban2012}, and \cite{Han2013}. Most recently, \cite{Gao2024} present a queuing model that exhibits MFD hysteresis. The seminal works of Luke \cite{Luke1972} and Newell \cite{Newell1993a}, \cite{Newell1993b} demonstrate how physical commonalities between queuing and kinematic wave models can be exploited to simplify the solution of the latter.

MFD hysteresis is not be confused with other phenomena in traffic research that bear the same name: Treiterer and Myers \cite{Treiterer1974} define hysteresis as the separation of speed-density curves into an accelerating and a decelerating branch ahead of traffic disturbances. Among others, Zhang \cite{Zhang1999} and Yeo and Skabradonis \cite{Yeo2009} offer theoretical explanations for this effect.

The rest of the article is structured as follows: Section\ref{sec:model} outlines the fundamentals of LWR theory and formulates models for road segment and incoming traffic. Section \ref{sec:qualitative} provides an analytical derivation of the qualitative time-dependent relationships. Section \ref{subsec:sensitivity} quantitatively examines the correlation between demand intensity and the extent of temporal asymmetry of macroscopic variables. Section \ref{sec:simulation} demonstrates the significance of theoretically derived hysteresis effects through realistic numerical examples and discusses control-relevant adjustments of the road geometry for reducing this undesirable phenomenon. Finally, Section \ref{sec:conclusion} compares the results with existing theoretical and empirical studies on MFD dynamics.

\section{Model and General Solution}
\label{sec:model}
The contemporary formulation of the LWR theory (\cite{ligwhi55}, \cite{ric56}) can be summarized as follows. If traffic progresses in the direction of increasing $x$ and $x_1 > x_2$, then the integral conservation of vehicles can be expressed as:
\begin{equation}
    \frac{\partial}{\partial t} \int_{x_1}^{x_2} k(x,t) \, dx + q(x_2,t) - q(x_1,t) = 0.
\end{equation}
Here, $q(x,t)$ is the flow rate of traffic at position $x$ and time $t$, connected to the cumulative flow $N(x,t)$ by:
\begin{equation}
    q(x,t) = \frac{\partial N}{\partial t}(x,t),
\end{equation}
The spatial derivative of the negative cumulative flow provides the density $k(x,t)$:
\begin{equation}
    k(x,t) = -\frac{\partial N}{\partial x}(x,t).
\end{equation}
If $k$ is differentiable, the conservation law can be expressed by differentiating $N$ with respect to time and space, leading to the partial differential equation:
\begin{equation}
    \frac{\partial k}{\partial t} + \frac{\partial q}{\partial x} = 0.
    \label{eq:conservation}
\end{equation}
In addition, the LWR theory assumes the existence of a fundamental relationship \( Q \), which might vary with location $x$ but not with time $t$:
\begin{equation}
	q(x,t)=Q(x,k(x,t)).
	\label{eq:fundamental}
\end{equation}
On substituting equation \ref{eq:conservation} into \ref{eq:fundamental}, we obtain 
\begin{equation}
    \frac{\partial k}{\partial t} + \frac{\partial Q}{\partial k} \cdot \frac{\partial k}{\partial x} = 0,
\end{equation}
which defines a unique solution for $k(x,t)$ and $q(x,t)$ for given initial and boundary conditions if $k$ is differentiable.

In cases where $k$ has a discontinuity at $(x,t)$, known as a shockwave, the shockwave's speed $u$ is specified as:
\begin{equation}
    u = \frac{[q]}{[k]} = \frac{q_2 - q_1}{k_2 - k_1}.
\end{equation}
Traffic flows through a link of length \( l \) with a continuous bottleneck starting at \( x_0 \) and extending until the end of the link at \( l \). The capacity of the bottleneck \( q_{bn}(x) \), for \( x_0 \leq x \leq l \), is a non-increasing function of \( x \) with \( q_{bn}(x_0) \) being equal to the capacity according to the fundamental diagram of the link. For notational convenience, we also define \( k_{bn} = \max \{ q^{-1}(q_{bn}) \} \).
Let \( A(t) \) be the accumulation, representing the total number of vehicles on the link at time \( t \). 
Let \( \bar{q}(t) \) be the average flow on the link at time \( t \), and let \( P(t) \) denote the flow at the downstream end at time \( t \). The trajectory of the tail of the queue in space-time is denoted by \( \psi(t) \). The upstream boundary flow $q(0,t)$ adheres to a trapezoidal, piece-wise linear function with maximum boundary flow $q_p$, i.e. 

\begin{equation}
	q(0,t) = 
	\begin{cases} 
		q_b + a \cdot t, & \text{for } 0 \leq t \leq \frac{q_p-q_b}{a}, \\
		q_p, & \text{for } \frac{q_p-q_b}{a} \leq t \leq t_{off}, \\
		q_p - b \cdot (t_{e} - t), & \text{for } t_{off} \leq t \leq \frac{q_p}{q_e \cdot b}+t_{off}, \\
		q_e, & \text{for } \frac{q_p}{q_e \cdot b}+t_{off} \leq t \leq \infty.
	\end{cases}
\end{equation}

for suitably chosen parameters \( q_b \) (initial flow), \( q_p \) (peak flow), \( q_e \) (end flow), $t_{off}$ (begin of the offset of congestion), \( a \) (flow increase rate at the onset of congestion), and \( b \) (flow reduction rate at the offset of congestion). We also define \( t_{\text{max}} := \arg \max_t \{ A(t) \} \).

The validity of our analysis does not require the corridor to provide sufficient space for the expanding queue; consequently, the queue may extend to position x=0. This represents the analytically degenerate, yet practically common case where the queue propagates beyond the position of the most upstream active detector. In this scenario, the flow at position x=0 reduces to the bottleneck flow, while all other properties of the solution remain unchanged.
 
In \ref{sec:general_solution}, we prove two lemmas that together yield a computationally convenient method for a general analytical solution of the model. It is not essential for the understanding of the following material and may be omitted if desired.

\section{Shapes of Hysteresis Loops}
\label{sec:qualitative}
We analyze how the figure-eight hysteresis pattern emerges in a link's Macroscopic Fundamental Diagram (MFD), starting with the simplest case: a single bottleneck at the downstream end of the link, during periods when there's an active queue behind this bottleneck.

Let's consider two moments in time, \(t_1\) and \(t_2\), where \(t_2\) comes after \(t_1\), and the total number of vehicles on the link (accumulation \(A\)) is identical at both times: \(A(t_1) = A(t_2)\). At \(t_1\), because the queue is growing, the flow just upstream of the queue must exceed bottleneck capacity. Conversely, at \(t_2\), as the queue is dissipating, the flow just upstream of the queue must be lower than bottleneck capacity. Within the queue itself, the flow remains constant at bottleneck capacity (\(q_{\text{bn}}(l)\)) at both times. Hence, the average flow at \(t_2\) upstream of the bottleneck is lower than the average at \(t_1\).

As shown in Figure \ref{figure:bn_schematic_mfd}, this relationship persists regardless of overlapping traffic states between times \(t_1\) and \(t_2\), and it extends to continuous bottlenecks.

\begin{figure}[H]
    \centering
    \begin{subfigure}[t]{0.45\textwidth}
        \centering
        \includegraphics[width=\textwidth]{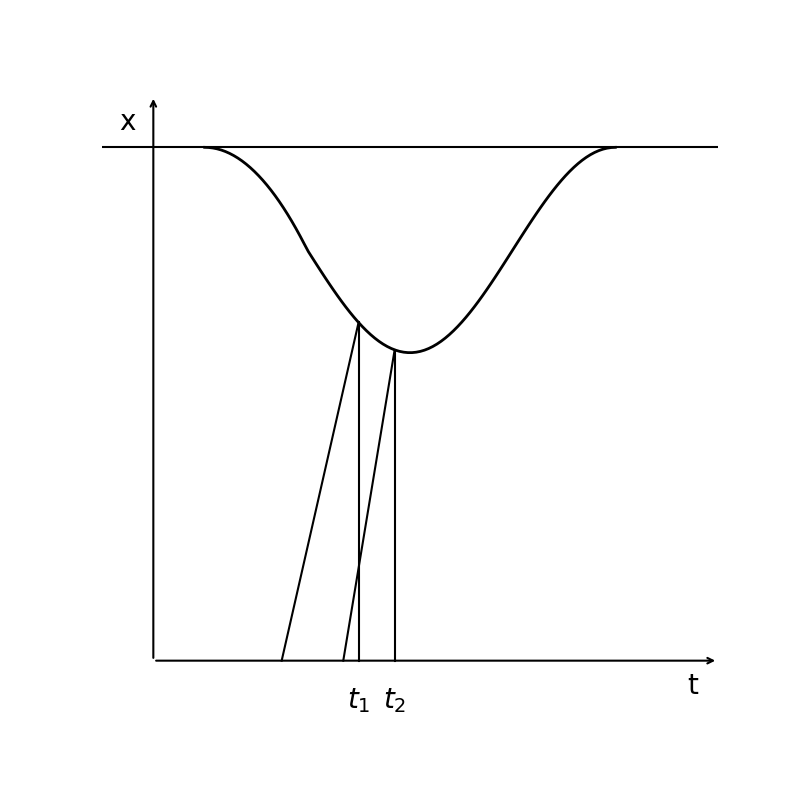}
        \caption{A scenario where the upstream local flow \(q(0, t_1)\) is smaller than the flow \(q(\psi(t_2), t_2)\). Nevertheless, the average at \(t_1\) is larger, since \(\bar{q}\) decreases in this interval while \(A\) remains constant.}
        \label{fig:1bn_overlapping}
    \end{subfigure}
    \hfill
    \begin{subfigure}[t]{0.45\textwidth}
        \centering
        \includegraphics[width=\textwidth]{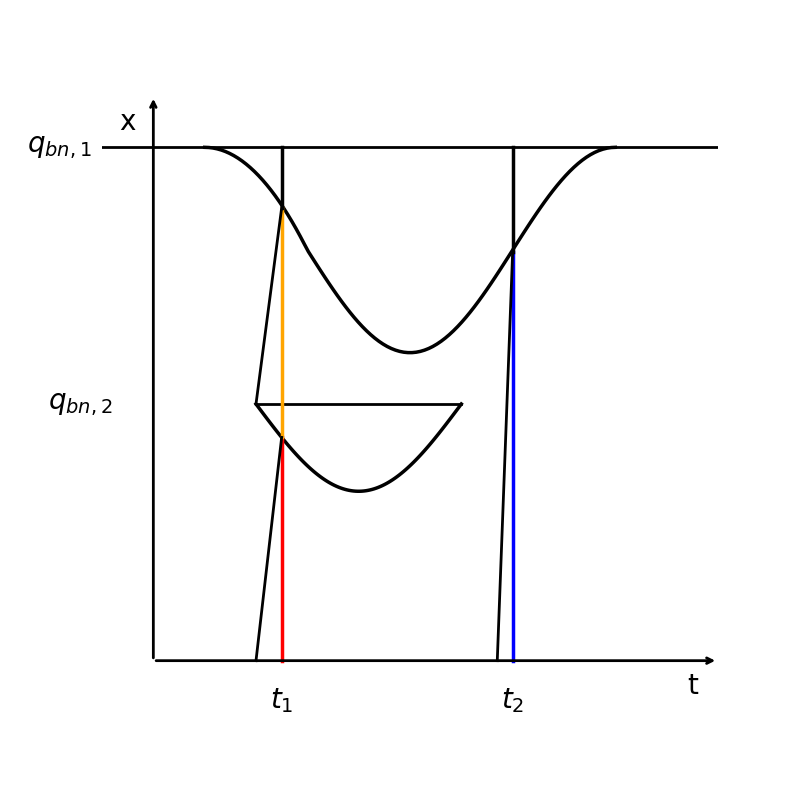}
        \caption{In the figure, for all points \((x,t)\) colored in red: \(q(x,t) > q_{\text{bn},2} > q_{\text{bn},1}\), for all points colored in orange \(q(x,t) = q_{\text{bn},2} > q_{\text{bn},1}\), and for all points colored in blue \(q_{\text{bn},1} > q(x,t)\). Since the minimum local flow at \(t_1\) is at least as high as the maximum flow at \(t_2\), the average at \(t_1\) is also higher.}
        \label{fig:2bn_schematic}
    \end{subfigure}
    \caption{Comparison of $\bar{q}$ between time points with equal accumulation}
    \label{figure:bn_schematic_mfd}
\end{figure}

When the period of observation is extended to the entire modeled interval, the shape of the curve is determined by the following theorem:

\begin{proposition}
\label{proposition:MFD}
The relationship between \( A(t) \) and \( \bar{q}(t) \) forms a figure-eight loop under the following conditions:
\begin{enumerate}
    \item The fundamental diagram is not a straight line in the considered interval.
    \item At least one of these must be true:
    \begin{enumerate}
        \item The inflow rate (\(a\)) is greater than the outflow rate (\(b\)), and by the end time (\(t_e\)), all congestion has cleared from the road section.
        \item The final flow rate (\(q_e\)) is higher than the initial flow rate (\(q_b\)).
    \end{enumerate}
\end{enumerate}

This figure-eight pattern can also show up in intervals of active queuing. If these conditions are not met, the relationship describes a simple clockwise loop.
\end{proposition}

The possibility of counter-clockwise partial loops is based on the fact that, if \(a > b\) holds, the distribution of vehicles across the corridor is more uniform during the offset of congestion, which at a fixed vehicle count implies a higher average flow. An intuitive proof approach is as follows: 

First, we introduce the following defintion.
\begin{definition}[Majorization, see \cite{Marshall2011}]
Let \(\mathbf{x} = (x_1, x_2, \ldots, x_n)\) and \(\mathbf{y} = (y_1, y_2, \ldots, y_n)\) be vectors in \(\mathbb{R}^n\). Define \(\mathbf{x}^\downarrow\) and \(\mathbf{y}^\downarrow\) as the vectors obtained by sorting \(\mathbf{x}\) and \(\mathbf{y}\) in descending order, respectively, such that:
\[
x^\downarrow_1 \geq x^\downarrow_2 \geq \cdots \geq x^\downarrow_n, \quad 
y^\downarrow_1 \geq y^\downarrow_2 \geq \cdots \geq y^\downarrow_n.
\]
We say that \(\mathbf{x}\) majorizes \(\mathbf{y}\), denoted by \(\mathbf{x} \succ \mathbf{y}\), if the following conditions hold:
\begin{enumerate}
    \item \(\sum_{i=1}^k x^\downarrow_i \geq \sum_{i=1}^k y^\downarrow_i, \quad \forall k = 1, 2, \ldots, n-1,\)
    \item \(\sum_{i=1}^n x^\downarrow_i = \sum_{i=1}^n y^\downarrow_i.\)
\end{enumerate}
\end{definition}

Next, we discretize the corridor into \(n\) cells, though the proof can be trivially extended to a continuous spatial variable. When at \(t_1\) and \(t_2\) no queue is active, and \(A(t_1) = A(t_2)\) and \(a > b\) hold, then \(q(1, t_1) \leq q(2, t_2)\) and \(q(0, t_1) \geq q(2, t_2)\) must hold. We discretize the corridor spatially into \(n\) cells. Then the following holds:

\begin{lemma}
The vector \(K_2:=\big(k(0, t_2), k(1, t_2), \ldots, k(n-1, t_2)\big)\) majorizes \(K_1 :=\big(k(0, t_1), k(1, t_1), \ldots, k(n-1, t_1)\big)\).
\end{lemma}

\begin{proof}
\(K_{1,i}^\downarrow\) is obtained by sorting \(K_1\) in the upstream direction, and \(K_{2,i}^\downarrow\) is obtained by sorting \(K_2\) in the downstream direction. For a contradiction, assume that there exists an \(m\) such that 
\[
\sum_{i=1}^m K_{1,i}^\downarrow > \sum_{i=1}^m K_{2,i}^\downarrow.
\]
For this inequality to hold, it must be the case that \(K_{1,m}^\downarrow \geq K_{2,m}^\downarrow\). Now let \(m' > m\). Then \(K_{1,m'}^\downarrow > K_{2,m'}^\downarrow\) holds, since, when moving by a constant step in space, the flow at \(t_1\) changes more strongly than at \(t_2\) (see Figure~\ref{figure:maj_demonstrative} for a representative numerical example). Consequently,
\[
\sum_{i=1}^n K_{1,i} > \sum_{i=1}^n K_{2,i}
\]
must hold, which contradicts the assumption.
\end{proof}

\begin{figure}[H]
    \centering
    \includegraphics[width=\linewidth]{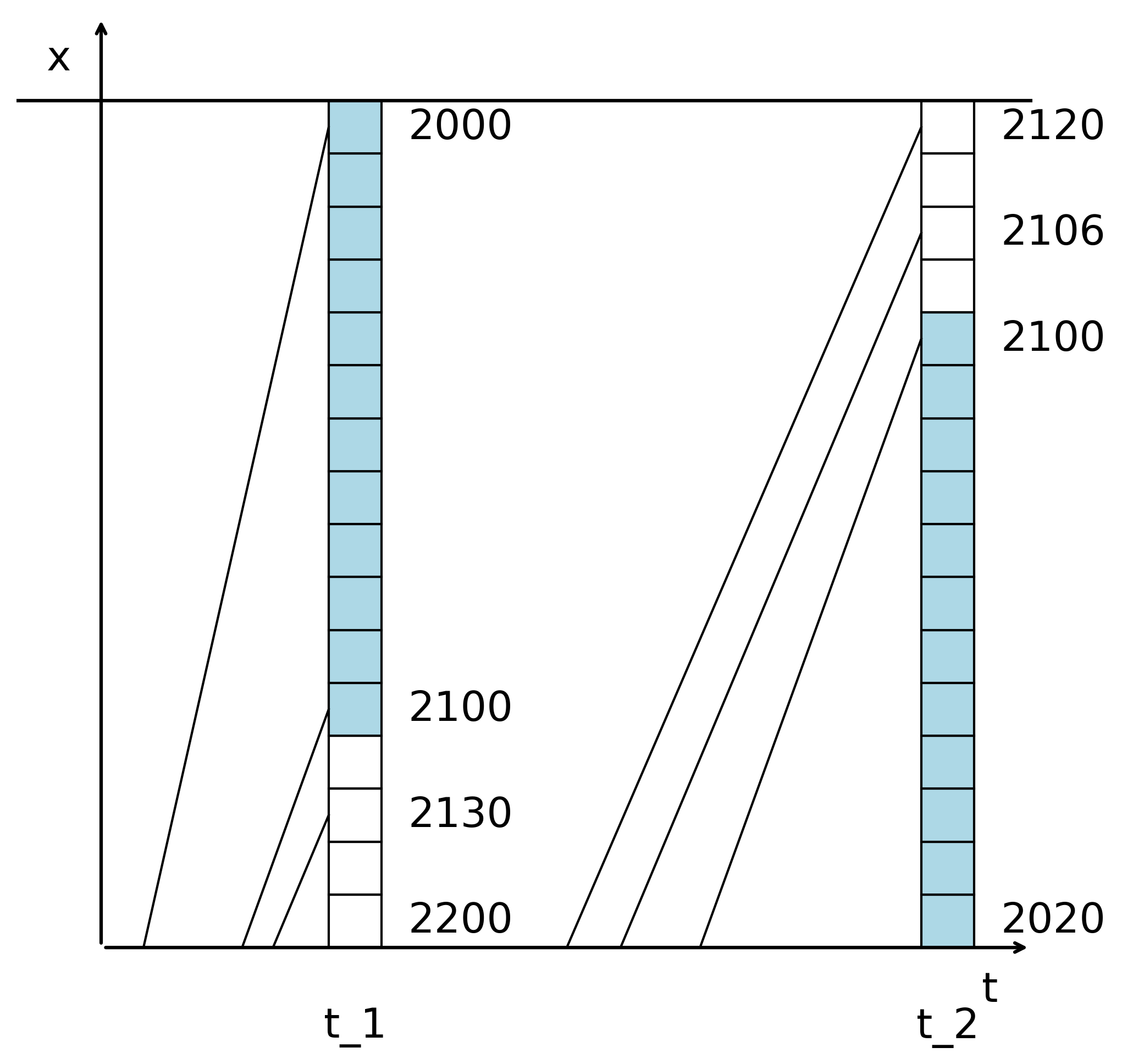}
    \caption{Straight lines in the figure represent characteristic waves along which the flow remains constant. For a contradiction, assume that the number of vehicles in the blue-marked cells at $t_1$ is higher than at $t_2$. This would require the flow in the respective final blue cell to be higher at $t_1$ than at $t_2$ (here 2110 vs 2100). Additionally, the vehicle count in the white cells would need to be lower at $t_1$. However, this is impossible since the flow in the corresponding downstream boundary region is consistently higher at $t_1$ (as indicated by flows of 2130 and 2106 respectively).}
    \label{figure:maj_demonstrative}
\end{figure}

Using this characterization, we can apply the following proposition from convex analysis to $\bar{q}(t_1)$ and $\bar{q}(t_2)$, which implies the desired result $\bar{q}(t_2) > \bar{q}(t_1)$.


\begin{lemma}[Karamata's inequality, see \cite{Marshall2011}]
Let $f$ be a concave function, and let $x_1, x_2, \ldots, x_n$ and $y_1, y_2, \ldots, y_n$ be two sequences of real numbers such that $x_1 \geq x_2 \geq \cdots \geq x_n$ and $y_1 \geq y_2 \geq \cdots \geq y_n$. If the sequence $\mathbf{x} = (x_1, x_2, \ldots, x_n)$ majorizes the sequence $\mathbf{y} = (y_1, y_2, \ldots, y_n)$, then:
\[
\sum_{i=1}^n f(x_i) \leq \sum_{i=1}^n f(y_i).
\]
\end{lemma}

By mapping the terms $k(0, t_1), \ldots, k(n-1, t_1)$ and $k(0, t_2), \ldots, k(n-1, t_2)$ to the sequences $\{x_i\}$ and $\{y_i\}$, respectively, and leveraging the concavity of $q$, we arrive at the conclusion:
\[
\bar{q}(t_2) > \bar{q}(t_1).
\]
If the fundamental diagram takes a linear form in the uncongested part, meaning $q(k)$ is not strictly concave, then the vehicle distribution in the uncongested part of the segment has no influence on the average flow, making counter-clockwise dynamics impossible. This is why Condition 1 needs to be included in the formulation of the proposition.

It remains to show that the pairs constructed in this way, where $\bar{q}(t_1) \leq \bar{q}(t_2)$ holds, can also occur when an active queue exists at both time points. We construct such a pair as follows: we choose $t_2$ as a point in time that lies arbitrarily shortly before the end of the most downstream, and thus longest-active, queue. Then there exists a $t_1$ that lies temporally before the onset of this queue and satisfies the required conditions. 

We modify the geometry by adding a further upstream queue whose capacity is chosen such that it begins arbitrarily shortly before $t_1$. The value of $A(t_1)$ does not change through this manipulation, and $\bar{q}(t_1)$ changes only marginally, so that all conditions mentioned above continue to be satisfied. For illustration, see Figure~\ref{figure:bn_counter}.

Finally, we need to show that after transitioning from the counterclockwise to the clockwise state, the system cannot revert to a counterclockwise regime through further increases in congestion. To analyze this, we compare the behavior when $A$ increases:

\begin{itemize}
    \item $t_1$ occurs during the onset of congestion, meaning $A$ increases as time progresses. Since the flow in the uncongested regime decreases with increasing $x$ during this interval, the waves entering the corridor have a higher average flow than the meeting the queue.
    \item $t_2$ occurs during the offset of congestion, so we must go backward in time to achieve an increase in $A$. The flow in the uncongested regime increases with increasing $x$ during this interval; therefore, the waves entering the corridor have a lower average flow than those moving into the queue.
\end{itemize}

\begin{figure}[H]
    \centering
    \begin{subfigure}[b]{0.45\textwidth}
        \centering
        \includegraphics[width=\textwidth]{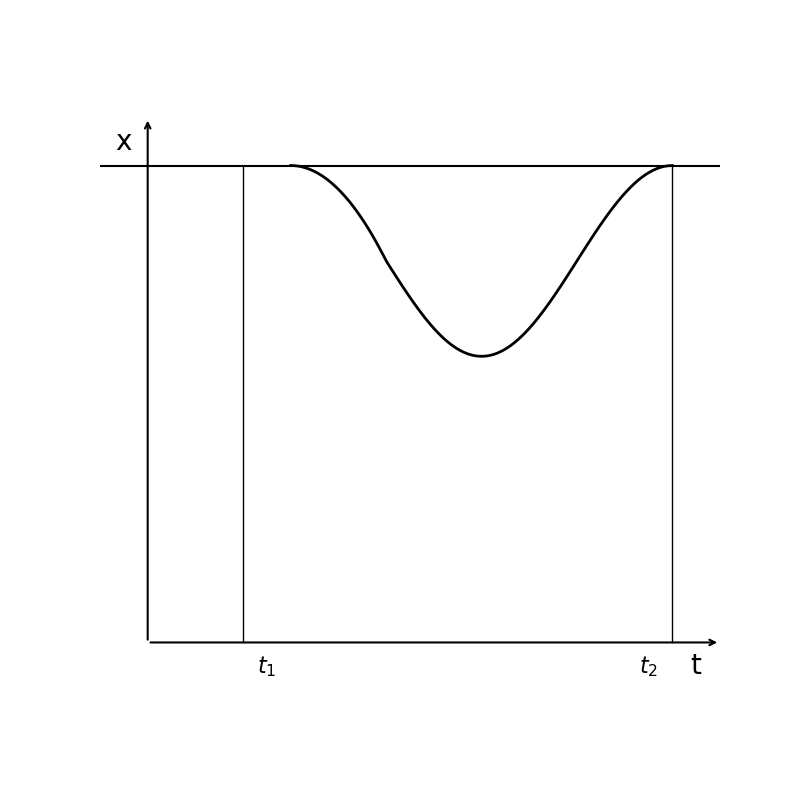}
        \caption{Corridor under initial conditions with $A(t_1)=A(t_2)$ and $\bar{q}(t_1)<\bar{q}(t_2$}
        \label{fig:1bn_counter}
    \end{subfigure}
    \hfill
    \begin{subfigure}[b]{0.45\textwidth}
        \centering
        \includegraphics[width=\textwidth]{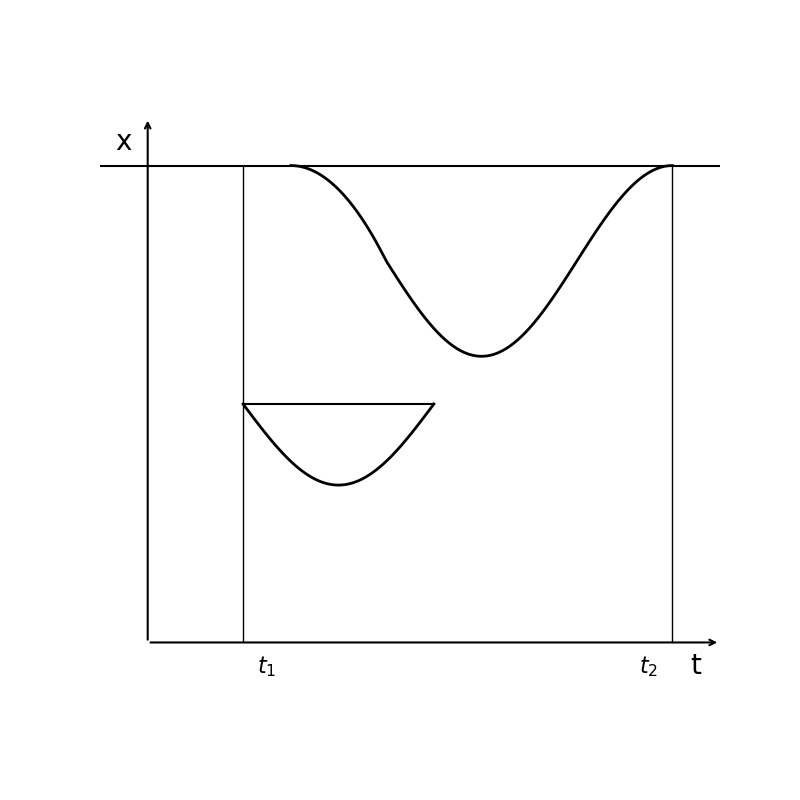}
        \caption{Both $t_1$ and $t_2$ experience active congestion, $A(t_1)=A(t_2)$ and $\bar{q}(t_1)<\bar{q}(t_2)$ still hold}
        \label{figure:2bn_counter}
    \end{subfigure}
    \caption{Construction of a counterclockwise loop during an interval of active queuing.}
    \label{figure:bn_counter}
\end{figure}

Therefore, starting from $t_1$, a smaller proportion of the increase in the number of vehicles contributes to the growth of the queue compared to $t_2$. As a result, the average flow increase is higher from $t_1$ than from $t_2$. Figure \ref{figure:acc_increase} illustrates this relationship with a representative numerical example. 

\begin{figure}[H]
    \centering
    \includegraphics[width=\linewidth]{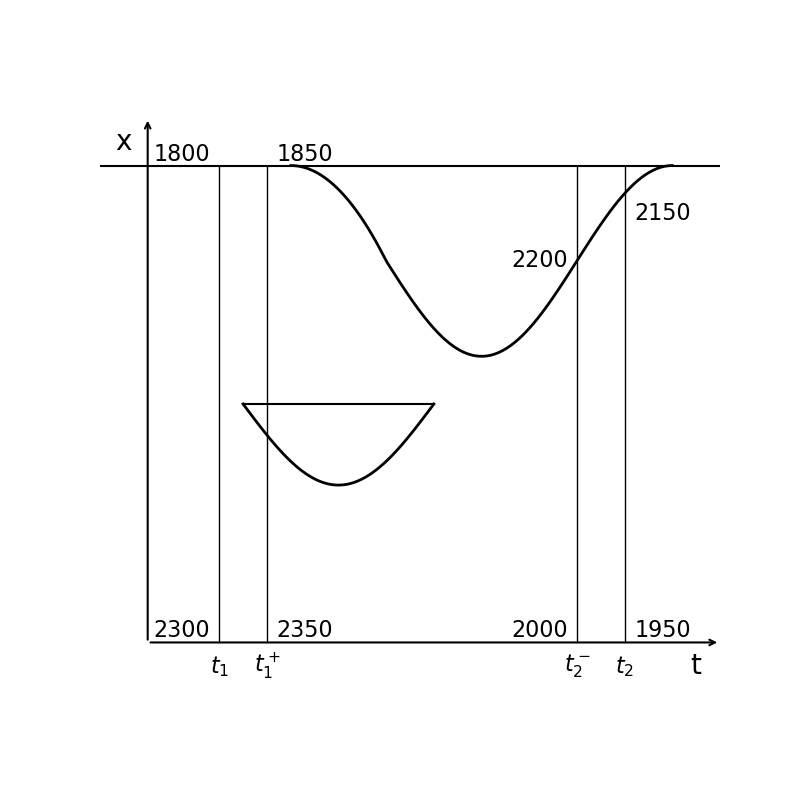}
    \caption{Since the queue grows more strongly between $t_2$ and $t_2^-$ than between $t_1$ and $t_1^+$, the following holds:
$A(t_1) = A(t_2)$, $\bar{q}(t_1) = \bar{q}(t_2)$, $A(t_1^+) = A(t_2^-)$, and $\bar{q}(t_1^+) > \bar{q}(t_2^-)$.}
    \label{figure:acc_increase}
\end{figure}

To finalize the proof of \ref{proposition:MFD}, two special cases require separate consideration:

\begin{itemize}
    \item $a \leq b$ and $q_e > q_b$: In this case, we can select $t_2$ very late, when the queue has completely receded and the flow is homogeneous across the entire corridor. Then we can find a $t_1$ such that $A(t_1) = A(t_2)$ and $t_1 < t_2$, where $t_1$ occurs during the onset of congestion, meaning the flow is not completely homogeneous at this time. From the concavity of $q$, it follows that $\bar{q}(t_2) >= \bar{q}(t_1)$, resulting in a figure-eight shape.
    
    \item $q_e \leq q_b$, and congestion persists at $t_e$: This case differs from those previously analyzed as the vehicle count at the end of congestion is lower than at $t = 0$. Thus, two time points $t_1$ and $t_2$ can only have the same vehicle count if $t_2 \leq t_e$. Under these conditions, $\bar{q}(t_1) \geq \bar{q}(t_2)$ always holds, and the resulting shape is a clockwise loop. This can be proven formally through the following chain of inequalities:

    \begin{align*}
\bar{q}(t_1) &= q\left(\frac{l - \psi(t_2)}{l} \left(\frac{1}{l - \psi(t_2)} \int_0^{l - \psi(t_2)} k(x, t_2) \, dx\right) + \frac{\psi(t_2)}{l} k_{bn}\right) \\
&\geq \left(\frac{l-\psi(t_2)}{l}\right) q\left(\frac{1}{l-\psi(t_2)} \int_0^{l-\psi(t_2)} k(x, t_2) \, dx\right) + \frac{\psi(t_2)}{l} q(k_{bn}) \\
&\geq \left(\frac{l-\psi(t_2)}{l}\right) \left(\frac{1}{l-\psi(t_2)} \int_0^{l-\psi(t_2)} k(x) \, dx\right) + \left(\frac{\psi(t_2)}{l}\right) q(k_{bn}) \\
&= \left(\int_0^{l-\psi(t_2)} q(k(x)) \, dx + \psi(t_2) q(k_{bn})\right)/l = \bar{q}(t_2),
\end{align*}

where the first inequality arises from the definition of concavity, and the second from Jensen's inequality for concave functions.
\end{itemize}

This concludes the proof of Proposition \ref{proposition:MFD}. The dynamics of the Network Exit Flow are described in the following proposition:

\begin{proposition}
The relationship between $A(t)$ and Network Exit Flow $P(t)$ always forms a counter-clockwise hysteresis loop.
\end{proposition}

The explanation is simple, we distinguish between the following two cases:

\begin{enumerate}
    \item \textbf{During queuing:} When congestion starts, accumulation exceeds the bottleneck capacity. When the queue dissolves, the same exit flow (at bottleneck capacity) occurs at a lower vehicle count than when congestion began.
    \item \textbf{During non-queuing:} Since the flow reaching the downstream end left the upstream end at an earlier point in time, and the boundary flow is decreasing during the onset of congestion and increasing during the offset of congestion, the Network Exit Flow (NEF) is higher during the offset at equal vehicle counts.
\end{enumerate}

\section{Demand Sensitivity}
\label{subsec:sensitivity}

In the following, we examine the impact of a change in peak demand on the area of the hysteresis loop \(H(q_p)\). To maintain generalizability of the results, we focus on the asymptotic growth of \(H\) and consolidate lower-order terms according to Landau notation. In this section, we assume a corridor with a single, discontinuous bottleneck at the downstream end. We focus on physically relevant cases of the model, where a traffic queue has already formed behind the bottleneck by the time the upstream flow begins to decrease, i.e. $\psi(t_{qe})<l$ holds.
To ensure appropriate scaling of the corridor to accommodate the given demand, we assume that scaling the peak demand by a factor $a$ results in both the flow $q(k)$ and the link length $l$ being multiplied by the same factor $a$. Under this assumption, we establish upper and lower bounds for the area $H$ under the MFD hysteresis curve and identify the shapes of the fundamental diagram $q(k)$ that correspond to these bounds. To provide a more convenient expression for the maximum extent of hysteresis, we additionally define the symbol $\tau$ for the free-flow travel time.

\begin{proposition}
\label{proposition:sensitivity}
For a given instance of the model, the following hold:
\begin{enumerate}[label=(\alph*)]
    \item The functional form that maximizes the area under the clockwise part of the MFD hysteresis curve is a triangular fundamental diagram with arbitrarily high jam density.
    \item The maximum area under this curve is a quadratic function of peak demand $q_p$, 
    \begin{align}
H(q_p) = & \, a^2 \left( \frac{(t_{pb} - \tau)^3}{3} - 0.5 \tau^2 t_{pb} + 0.5 \tau^3 \right) \notag \\
& + q_p^2 \left( \frac{\tau}{2} + (t_{pe} - t_{pb} + \tau)^2 + t_e - t_{pe} \right) \notag \\
& + b^2 \left( \left(\frac{1}{6} - 0.25 \right)\tau^3 - \frac{\tau^2}{3} + \left(\frac{1}{3} + 0.25 \tau \right)(t_e - t_{pe})^2 \right) \notag \\
& + b q_p \left( -(t_e - t_{pe})^2 + 0.5 \tau t_{pe} + 0.5 \tau t_e + \frac{\tau^2}{2} \right)+\Theta(q_p). \label{eq:Hqp}
\end{align}
\end{enumerate}
\end{proposition}

\begin{proof} See Appendix \ref{section:appendix_proofs}.
\end{proof}

The lower bound for the area under the hysteresis curve is zero. For instance, this occurs when the fundamental diagram increases linearly up to \( q = q_{bn} + \epsilon_1 \), and for \( q > q_{bn} + \epsilon_1 \), the slope becomes \( dq/dk = \epsilon_2 \), where \( \epsilon_1 \) and \( \epsilon_2 \) are arbitrarily small values. The slope of the linear segment must be selected appropriately to satisfy the previously defined consistency conditions.
 This configuration of the fundamental diagram results in states where \( q(x,t) > q_{bn} + \epsilon_1 \) occupy only a marginal spatial portion of the corridor and therefore have no influence on the spatial average of the flow. Thus, \( \bar{q}(t) \) is arbitrarily close to \( q_{bn} \) for the entire interval of active congestion behind the bottleneck.

The derivation of upper and lower bounds NEFs follows similar principles. Since $P(t) = q_\text{bn}$ always holds in the interval of active congestion, the two-dimensional area described by the NEF curve inadequately reflects the temporal asymmetry in the relationship between outflow and accumulation. Therefore, we use the difference between maximum and minimum accumulation at which the maximum outflow $q_\text{bn}$ is reached as a metric for hysteresis in the NEF diagram. The properties of the solution are summarized in the following proposition:

\begin{proposition}
    \label{proposition:NEF_max}
The fundamental diagram that maximizes hysteresis in the NEF is characterized by the following density-flow relationship \( k(q) \):
\[
k(q) = 
\begin{cases} 
0 & \text{for } 0 \leq q \leq q(0, t_{\text{min}}) \\
(t_{\text{pb}} - t_{\text{bn}})(q - q(0, t_{\text{min}})) & \text{for } q(0, t_{\text{min}}) \leq q
\end{cases}
\]
Under these conditions, the value of the hysteresis metric is given by:
\[
A(t_{\text{max}}) - A(t_0) = A(t_{\text{max}}) = N(0, t_{\text{max}}) - N(l, t_{\text{pb}}) - (t_{\text{max}} - t_{\text{pb}})q_{\text{bn}}
\]
\end{proposition}

The proof for this statement is also provided in appendix \ref{section:appendix_proofs}.

On the other hand, the minimum value of the NEF hysteresis metric is 0, which is reached, for example, when \( \frac{dk}{dq} = 0 \) for all \(q\) in uncongested conditions, as the system then behaves like a conventional queue.

\section{Empirical Validation}
\label{sec:empirical}

\subsection{Bottleneck Geometry and MFD Hysteresis}
In this section, we empirically validate the hypothesis that a bottleneck with gradually decreasing geometry generates less hysteresis under otherwise identical conditions.
In terms of traffic flow theory, this statement is straightforward: First, we note that the number of vehicles at any given time is independent of the bottleneck geometry, as capacity is always lowest at the most downstream position. Consider two points in time \(t_1\) and \(t_2\) (with \(t_2 > t_1\)) with identical vehicle counts—first in a corridor with a single bottleneck, and then analyzing how traffic behavior changes by adding a second bottleneck at position \(x_1\).

When adding the additional bottleneck, vehicle accumulation between $x=0$ and $x=x_1$ increases more rapidly at $t_1$ than at $t_2$. This occurs because the preliminary bottleneck's higher capacity leads to both a faster reaching of peak vehicle density in this road segment and a higher discharge rate at subsequent times compared to that of the total accumulation. From this, we can conclude that the spatial proportion of the section with uncongested conditions decreases more strongly in the first part. Consequently, the spatial average of the flow at time \(t_1\) decreases more strongly than at \(t_2\).

Through iterative application, this argument formulated for two bottlenecks can be extended to roads with arbitrary geometric configurations.

To illustrate the theoretical explanation, we conduct a simulation using the Cell Transmission Model (CTM). The model instance comprises an initially empty 15 km road segment over a four-hour time interval. The inflow develops as follows: Starting from zero, it increases linearly to 3000 vehicles per hour at t=30 minutes, then decreases linearly back to zero at t=4 hours.
Traffic flow dynamics are described by a triangular fundamental diagram with the following parameters: free-flow speed 120 km/h, jam density 250 vehicles/km, and capacity 6000 vehicles per hour. We examine two scenarios: The first case features a single bottleneck at the downstream end with a capacity of 1800 vehicles per hour. In the second case, an additional bottleneck with a capacity of 2100 vehicles per hour is placed at the midpoint of the road.
The traffic dynamics of both scenarios are visualized in Figure \ref{figure:demo} as a heat map. Red vertical lines mark two representative time points (t=75 and t=153) with identical vehicle accumulation. It clearly shows that the blue region, indicating uncongested traffic conditions, diminishes more strongly at the earlier time point due to the additional bottleneck. The area under the hysteresis curves depicted in figure \ref{figure:demo_mfd} measures 25123 in the first case and reduces to 20374 in the second case.

\begin{figure}[H]
    \centering
    \begin{subfigure}[b]{0.45\textwidth}
        \centering
        \includegraphics[width=\textwidth]{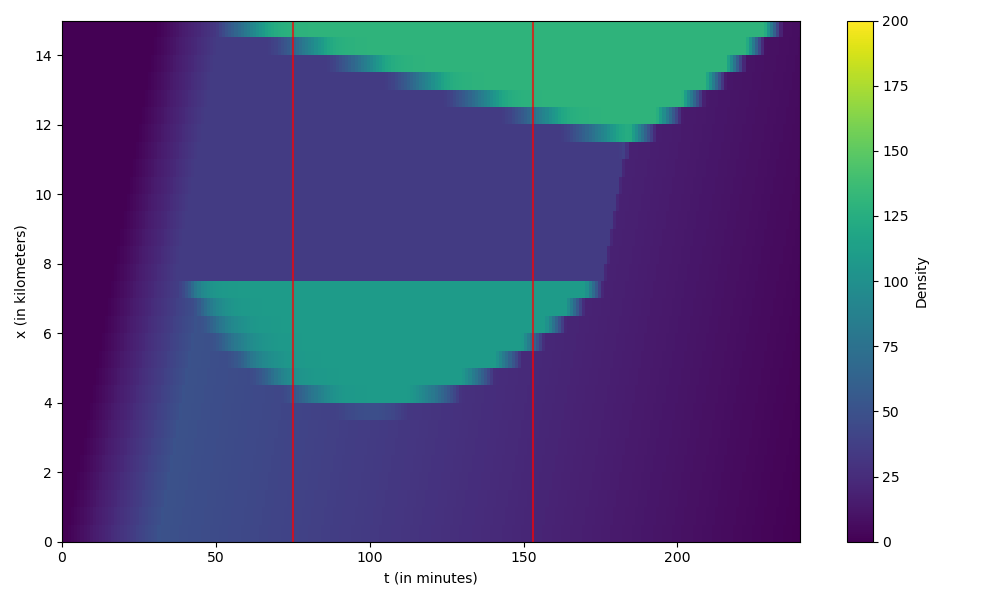}
        \caption{Single Bottleneck}
    \end{subfigure}
    \hfill
    \begin{subfigure}[b]{0.45\textwidth}
        \centering
        \includegraphics[width=\textwidth]{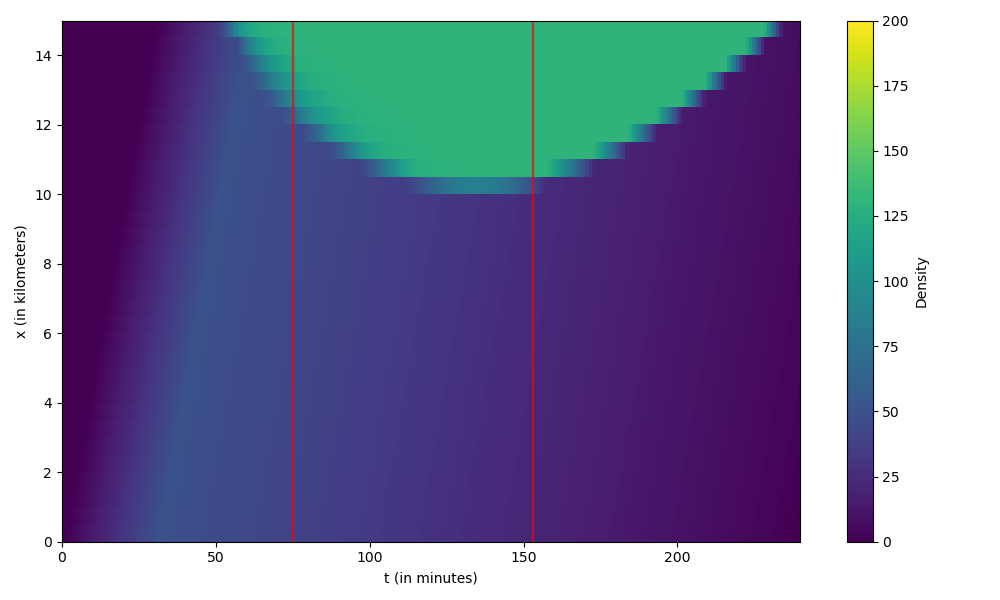}
        \caption{Two Bottlenecks}
    \end{subfigure}
    \caption{Heatmaps of traffic densities in both simulated scenarios}
    \label{figure:demo}
\end{figure}

\begin{figure}[H]
    \centering
    \begin{subfigure}[b]{0.45\textwidth}
        \centering
        \includegraphics[width=\textwidth]{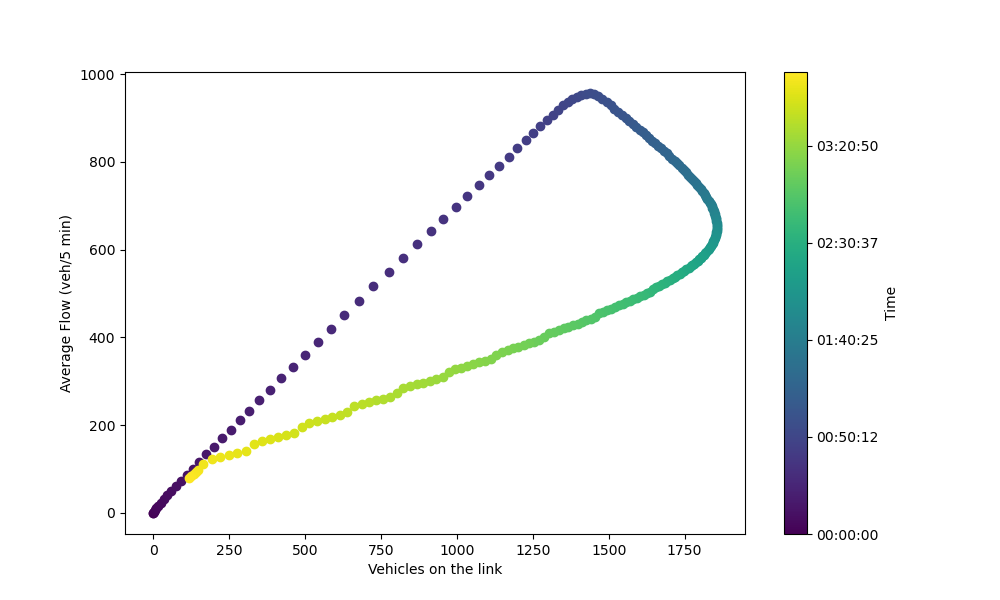}
        \caption{Single Bottleneck}
    \end{subfigure}
    \hfill
    \begin{subfigure}[b]{0.45\textwidth}
        \centering
        \includegraphics[width=\textwidth]{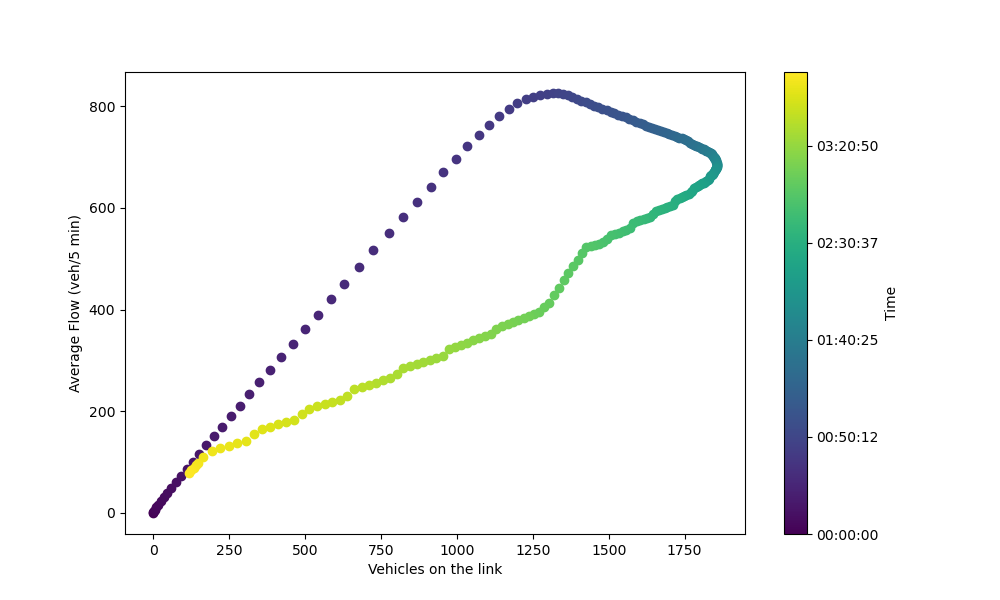}
        \caption{Two Bottlenecks}
    \end{subfigure}
    \caption{MFDs with Clockwise Hysteresis Loops}
    \label{figure:demo_mfd}
\end{figure}

\subsection{The Study Sites}
The empirical analysis is based on measurements from two representative road segments. The data was obtained from the Performance Measurement System (PeMS) of the California Department of Transportation (Caltrans). Data collection occurred on 29 workdays between February 5, 2024, and March 17, 2024, with traffic flow and occupancy rates aggregated in 5-minute intervals. Following the convention of traffic flow research (see e.g. \cite{gersun11}, \cite{builad09}), average flow and occupancy were calculated as the unweighted average of individual detector measurements.
The first study site comprises a 7.4 km corridor section of Interstate 880 northbound in the San Francisco metropolitan area, immediately upstream of the bottleneck at the Washington Avenue off-ramp. Data collection was conducted using 18 detectors along the corridor. The traffic node geometry and congestion formation mechanisms are documented in detail in studies [1], [2], and [3]. According to [2], the oversaturated exit leads to congestion on all I-880 lanes extending approximately one kilometer upstream of the exit point. Downstream of this point, the road capacity increases gradually, allowing the bottleneck to be classified as discontinuous. 

The MFD for this road section over the course of a day is presented in Figure \ref{figure:wash_full}. Individual data points were calculated by taking the arithmetic mean of the corresponding times across all measurement days. Two clockwise hysteresis loops are clearly visible, with a slightly higher level of congestion during the evening peak hour. An additional partial counter-clockwise loop occurs between increasing traffic in the early morning hours (starting around 4 AM) and decreasing traffic in the late evening hours (until approximately 11 PM). As explained in the previous sections, this phenomenon can be attributed to the fact that the rate of flow increase at the upstream corridor end during early morning hours significantly exceeds the absolute value of the evening traffic decrease rate. Linear interpolation of the time-dependent traffic flow at this location yields a slope of 1.65 for the period from 4 to 8 AM, and a slope of -0.77 for the period from 7 to 11 PM (cf. Figure \ref{figure:upstr_boundary}).The dynamical behavior shown in Figure \ref{figure:wash_full} differs from our analysis in Section \ref{sec:qualitative} in that it captures two peak traffic periods - morning and evening - which manifests in two clockwise hysteresis loops. When examining the morning peak period in isolation, such as the interval from 6 to 10 AM shown here, a single hysteresis loop emerges, as illustrated in Figure \ref{figure:washav_morning}.

The second study site extends over 6 km upstream and immediately downstream of the Ashlan Avenue on-ramp to Highway 41 (northbound) near Fresno, California. At this location, 11 detectors were used for data collection. The congestion dynamics at this location are described in source \cite{freeway41}. The description in the article and the geometric configuration of this section indicate that during peak traffic hours, vehicles merging from the shoulder lane progressively congest one lane after another. This continuous lateral propagation of the congested area allows for the classification of this bottleneck as a continuous bottleneck. For the second study site, Figure \ref{figure:washav} presents the MFD for both the entire day and the morning peak period from 6 to 10 AM, with data points calculated as arithmetic means across all measurement days. In this case, the full-day diagram exhibits two overlapping clockwise loops, without any counter-clockwise components. Again, the morning peak period shows a single clockwise hysteresis loop.

\subsection{Comparison of Locations}
We generate MFD curves for both locations for the period between 6 and 10 AM on all 29 workdays between February 5, 2024, and March 15, 2024. The hysteresis curves were smoothed using a Savitzky-Golay filter \cite{Savitzky1964}, applying this method separately to the time series of flow and occupancy to preserve temporal correlation while reducing measurement noise. Figure \ref{figure:boundary_flows} visualizes the daily progression of traffic flow at the most upstream detectors of both study sites, with data smoothed using a moving average. The vehicle inflow at the I-880N site significantly exceeds that of SR-41N. Conversely, the occupancy at the Ashlan Avenue on-ramp exhibits higher maximum occupancy, indicating a lower minimum bottleneck capacity at this location. Both bottleneck capacity and boundary flow influence the extent of hysteresis in addition to the geometric form of the bottleneck, whose effect we aim to examine in isolation. To compensate for these different bottleneck capacities and boundary flows, the area under the smoothed curve was normalized by dividing by the maximum occupancy of the respective day.

Assuming a triangular fundamental diagram and sroughly equal times of congestion formation and dissipation, this can be justified theoretically as follows: If the capacity of the bottleneck at the first site ($q_{\text{bn,1}}$) is $p$ percent higher with identical congestion onset, then due to the triangular shape of the fundamental diagram, the accumulation or occupancy at this time is also $p$ percent higher than at the second site.

The total number of queued vehicles results from the piecewise linear boundary condition as:
\[
N_{\text{bn}}(q_{\text{bn}}) = 0.5 \cdot (t_{2,q_{\text{bn}}} - t_{1,q_{\text{bn}}}) \cdot (q_p - q_{\text{bn}}),
\]
where $t_{2,q_{\text{bn}}}$ and $t_{1,q_{\text{bn}}}$ represent the times with boundary flow $q_{\text{bn}}$ in descending order. From the assumptions made, it follows that $t_{1,q_{\text{bn}}}$ is independent of the peak flow value. A $p$ percent increase in $q_{\text{bn}}$ necessitates a corresponding increase in $q_p$, resulting in:
\[
N_{\text{bn}}(p \cdot q_{\text{bn}}) = p \cdot N_{\text{bn}}.
\]

At any time $t_0$ during congestion formation, the flow in the area in which both corridors are congested is $p$ percent higher at study site 1, corresponding to the difference in bottleneck capacities. This also applies to the uncongested area. Additionally, due to lower congestion density with higher vehicle count $N_{\text{bn}}$, there exists an area that is congested only at the first site at time $t_0$. Here, the flow at the first site equals $q_{\text{bn,1}}$, while at the second site it exceeds $q_{\text{bn,2}} = q_{\text{bn,1}} / p$ during congestion formation and falls below $q_{\text{bn,2}}$ during congestion dissolution. This situation is illustrated graphically in Figure \ref{figure:2waves}.

Consequently, during congestion formation, the average flow at the second site is higher than $1/p$ times the average flow at the first site, while it is lower at equal occupancy during congestion dissolution. Thus, the vertical component of the hysteresis area is larger at the second site, while the horizontal component is identical after normalization, as both start and end times and accumulations with hysteresis are scaled by factor $p$. The different corridor lengths do not affect this result, as we measure percentage occupancy rather than total accumulation on the $x$-axis.

This theoretical traffic flow analysis suggests that, assuming negligible effects from bottleneck geometry, the normalized hysteresis area at the second location (\(H_2\)) should be at least as large as that at the first location (\(H_1\)) due to a lower boundary flow. However, empirical data reveal mean hysteresis values of 35.14 and 24.65 for the two study sites, respectively—a reduction of 29.87\%. A two-sample \(t\)-test yields a \(p\)-value of 0.0111, providing sufficient evidence to reject the null hypothesis (\(H_2 \geq H_1\)) at all common significance levels. Additional statistical parameters of the datasets are summarized in Table \ref{table:stat_params}. These findings suggest that continuous bottleneck geometries also diminish MFD hysteresis in empirical scenarios. To determine the unit of area, the vehicle count can be considered dimensionless as it represents a counted value. Since occupancy is also dimensionless, at a 5-minute aggregation interval, the area under the curve has the unit \SI{0.2}{\per\minute} for both the raw and normalized area, as is consistent with a cumulative difference in average flows.

\begin{table}[H]
    \centering
    \begin{tabular}{|c|cc|cc|}
        \hline
        \textbf{Study Site} & \multicolumn{2}{c|}{\textbf{Raw Area}} & \multicolumn{2}{c|}{\textbf{Normalized Area}} \\ \hline
        & \textbf{Mean} & \textbf{Std. Dev.} & \textbf{Mean} & \textbf{Std. Dev.} \\ \hline
        I-880 N & 534.78 & 326.37 & 35.14 & 20.25 \\ \hline
        SR-41 N & 448.62 & 224.57 & 24.65 & 13.74 \\ \hline
    \end{tabular}
    \caption{Statistical parameters of the distribution of areas under the daily hysteresis curves}
    \label{table:stat_params}
\end{table}

\begin{figure}[H]
    \centering
    \includegraphics[width=\textwidth]{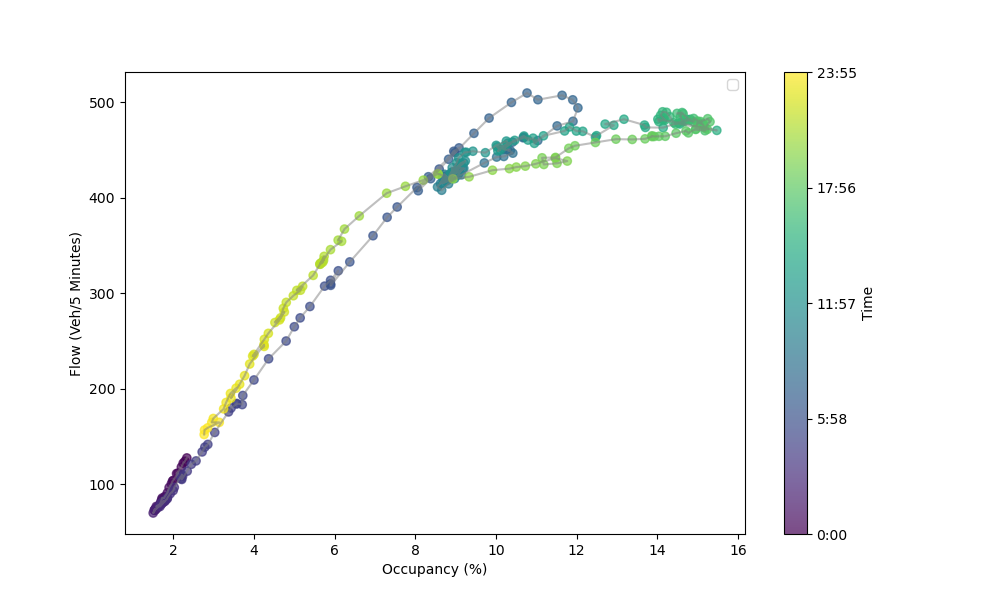}
    \caption{Flow-density relationship over an entire day, I-880N south of the Washington Avenue off-ramp}
    \label{figure:wash_full}
\end{figure}

\begin{figure}[H]
    \centering
    \includegraphics[width=\textwidth]{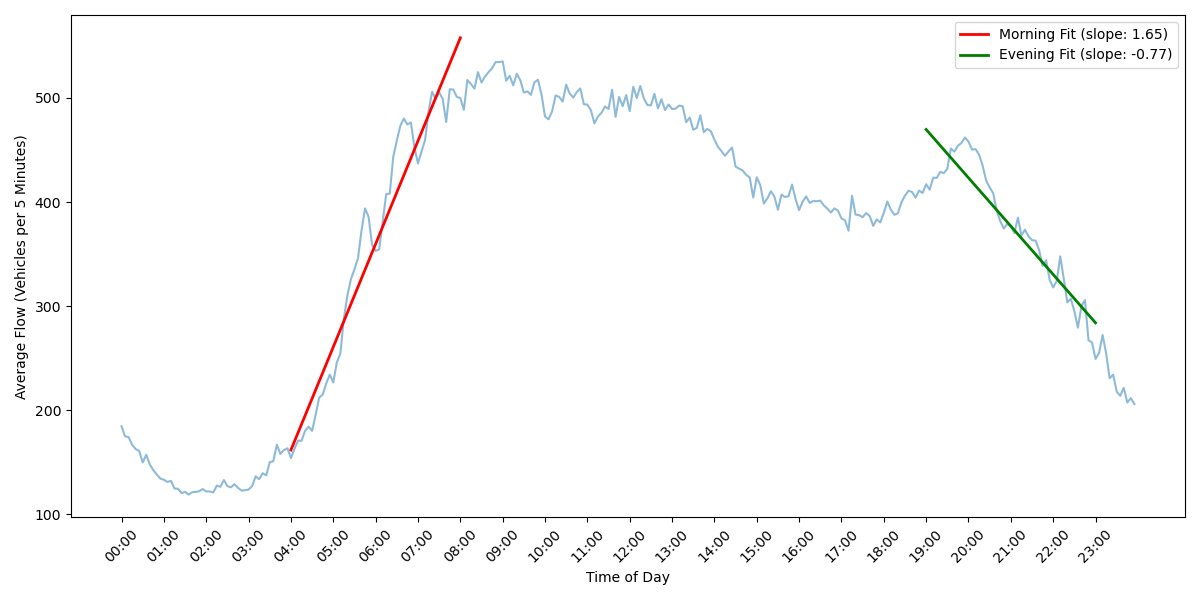}
    \caption{Flow-time diagram at the upstream end of the corridor with regression lines for the intervals 4 a.m. to 8 a.m. and 7 p.m. to 11 p.m.}
    \label{figure:upstr_boundary}
\end{figure}

\begin{figure}[H]
    \centering
    \includegraphics[width=\textwidth]{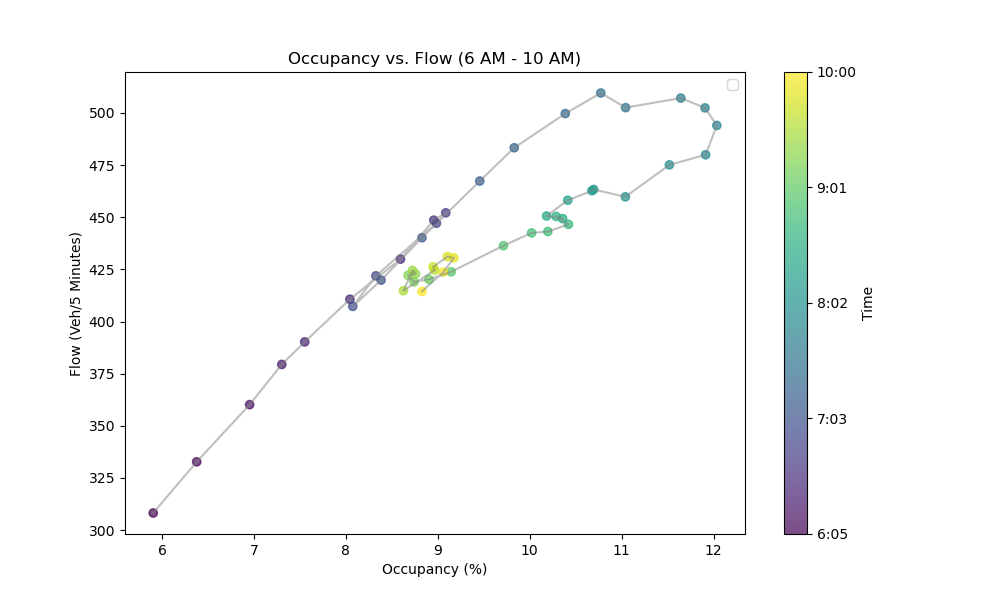}
    \caption{Flow-density relationship over at I-880N, 6 a.m. to 10 a.m.}
    \label{figure:upstr_boundary}
\end{figure}

\begin{figure}[H]
    \centering
    \begin{subfigure}[b]{0.49\textwidth}
        \centering
        \includegraphics[width=\textwidth]{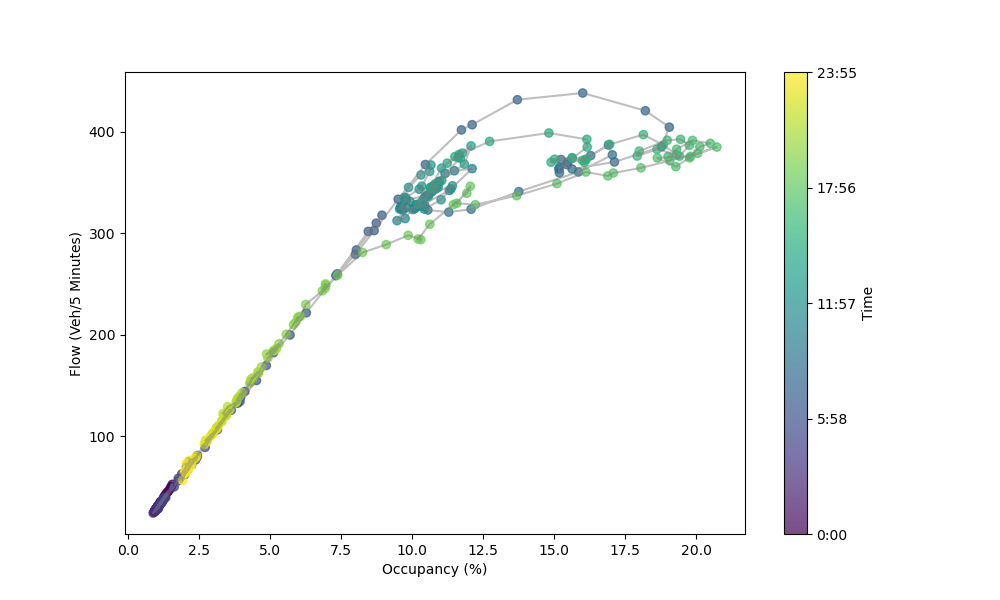}
        \caption{Full-day}
        \label{fig:ashav_full}
    \end{subfigure}
    \hfill
    \begin{subfigure}[b]{0.49\textwidth}
        \centering
        \includegraphics[width=\textwidth]{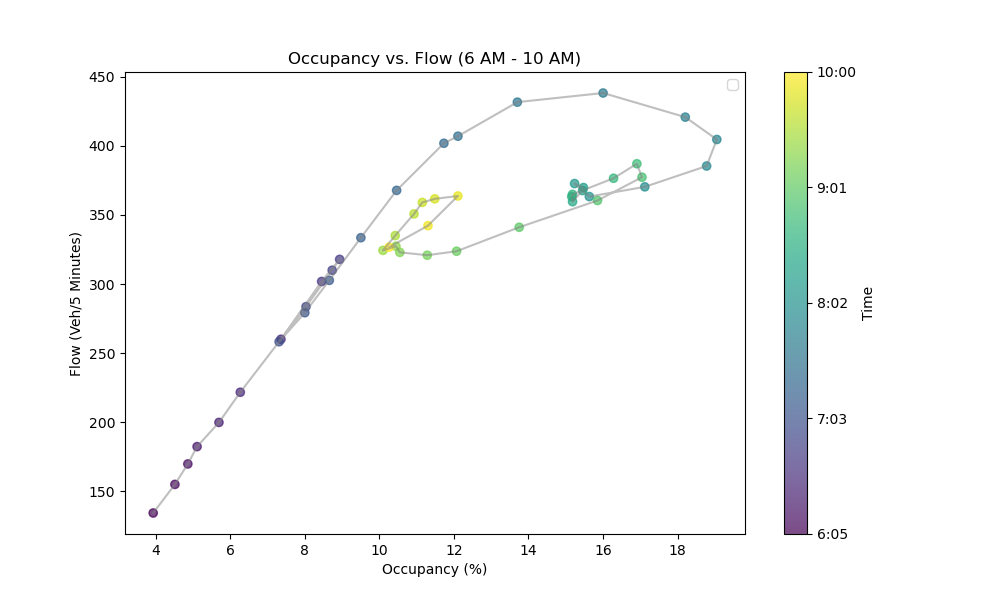}
        \caption{Morning peak (6 a.m. to 10 p.m.)}
        \label{fig:ashav_morning}
    \end{subfigure}
    
    \caption{MFDs of SR-41 in northbound direction near Ashlan Avenue on-ramp.}
    \label{figure:washav}
\end{figure}

\begin{figure}[H]
    \centering
    \includegraphics[width=\textwidth]{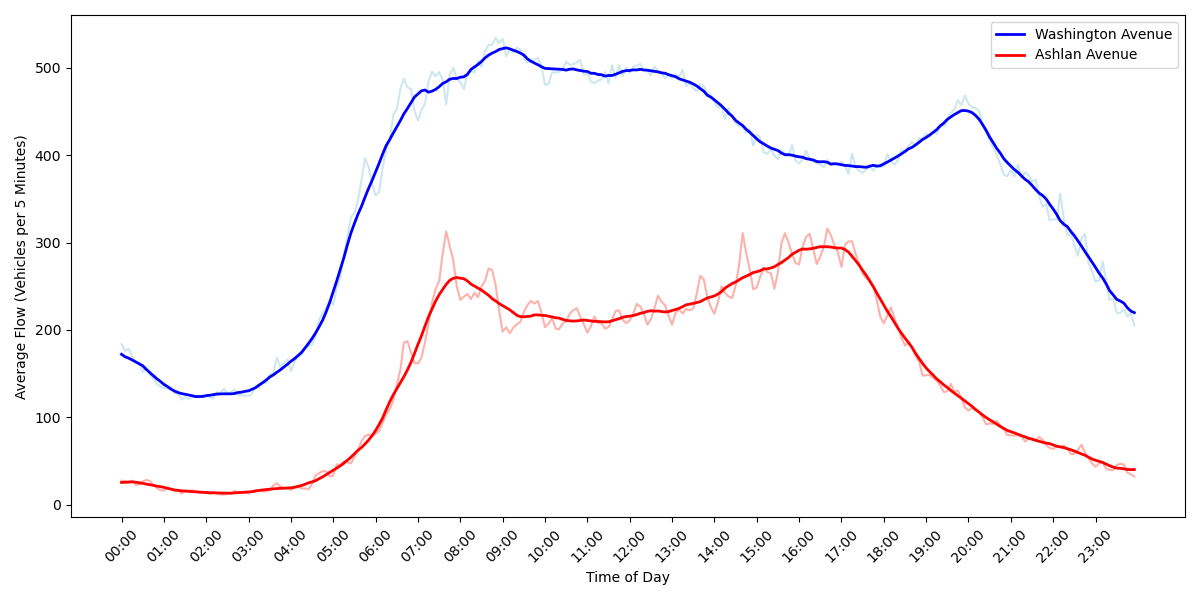}
    \caption{Local flow measured at the most upstream detectors at both study sites.}
    \label{figure:boundary_flows}
\end{figure}

\begin{figure}[H]
    \centering
    \includegraphics[width=\linewidth]{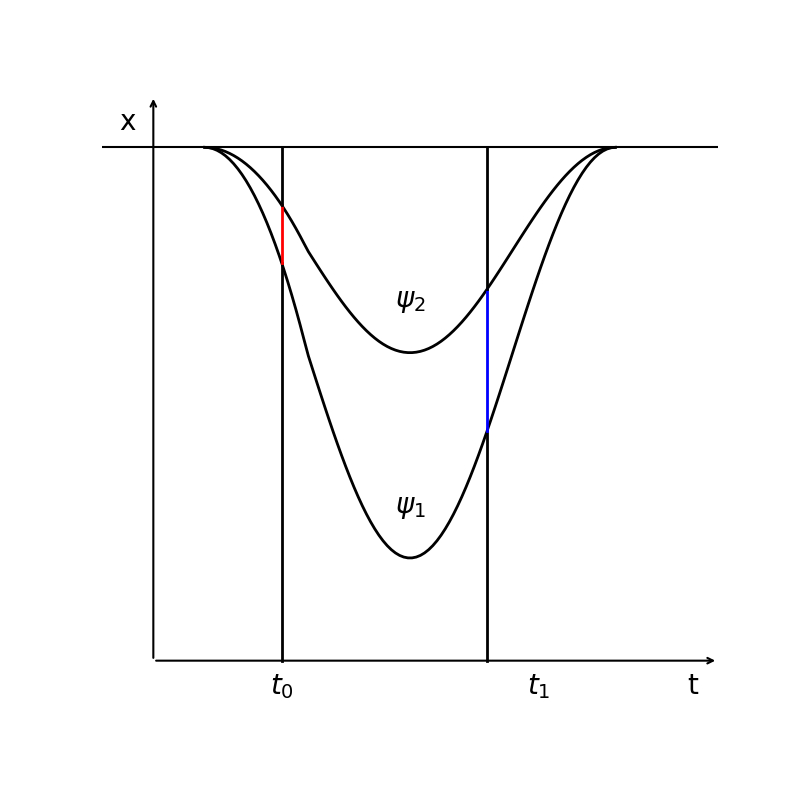}
    \caption{The curves $\psi_1$ and $\psi_2$ represent the trajectories of the tails of the queues at study locations 1 and 2, respectively. Consider two arbitrary time points $t_0$ and $t_1$ with equivalent accumulation values. Let $q_1(x,t)$ denote the flow rate at the first study location and $q_2(x,t)$ at the second location. In regions where the vertical lines are colored in black, it holds that $q_1(x,t) = (1+p)q_2(x,t)$. The red-colored portion indicates where $q_1(x,t) \leq (1+p)q_2(x,t)$, while the blue-colored portion denotes where $q_1(x,t) \geq (1+p)q_2(x,t)$.
}
    \label{figure:2waves}
\end{figure}

\section{Numerical Experiments}
\label{sec:simulation2}

\subsection{Bottleneck Geometry and MFD Hysteresis}

In this section, we conduct traffic simulations using the Cell Transmission Model (CTM) \cite{Daganzo1994} to analyze the effects of various configurations on hysteresis dynamics. Our parameter selection is based on data from I-880 North, the more congested of the two studied road sections. Since this heavily trafficked segment maintains significant congestion levels during midday hours, we focused our parameter estimation on the evening rush hour period (6:00 PM to 10:00 PM). To estimate the upstream boundary condition, we approximate the time-flow measurements from the uppermost detector presented in the previous chapter using the following piecewise linear function \( f(t) \):
\[
f(t) =
\begin{cases} 
at + b & \text{for } t \leq t_p, \\
ct + d & \text{for } t > t_p,
\end{cases}
\]
where t represents hours after 6 PM (e.g., t = 2.5 for 8:30 PM) subject to the constraints:
\[
\begin{aligned}
a &\geq 0 \quad \text{(non-negative slope in first segment)}, \\
b &\geq 0 \quad \text{(non-negative intercept)}, \\
c &\leq 0 \quad \text{(negative slope in second segment)}, \\
at_p + b &= ct_p + d \quad \text{(continuity at breakpoint \( t_p \))}.
\end{aligned}
\]

The parameters \((a, b, c, d, t_p)\) were estimated simultaneously using Sequential Least Squares Programming (SLSQP). The estimation yields:
\[
f(t) =
\begin{cases} 
447.23t + 5571.84 & \text{for } t \leq \text{8:00 PM}, \\
-620.37t + 7708.81 & \text{for } t \geq \text{8:00 PM}.
\end{cases}
\]

The bottleneck capacity is determined using measurements from two representative detectors centrally positioned within the corridor. Average flow and occupancy are calculated for each 5-minute interval during the morning peak period, and the resulting hysteresis curve is smoothed using a Savitzky-Golay filter. The resulting characteristics are shown in Figure~\ref{fig:estimate_bn_flow} and clearly demonstrate that both detectors are significantly affected by the queue. The simulated bottleneck capacity is selected as the flow at the rightmost point of the smoothed curve, corresponding to a rounded value of \(q_{\text{bn}} = 6240\).

\begin{figure}[H]
    \centering
    \includegraphics[width=0.8\textwidth]{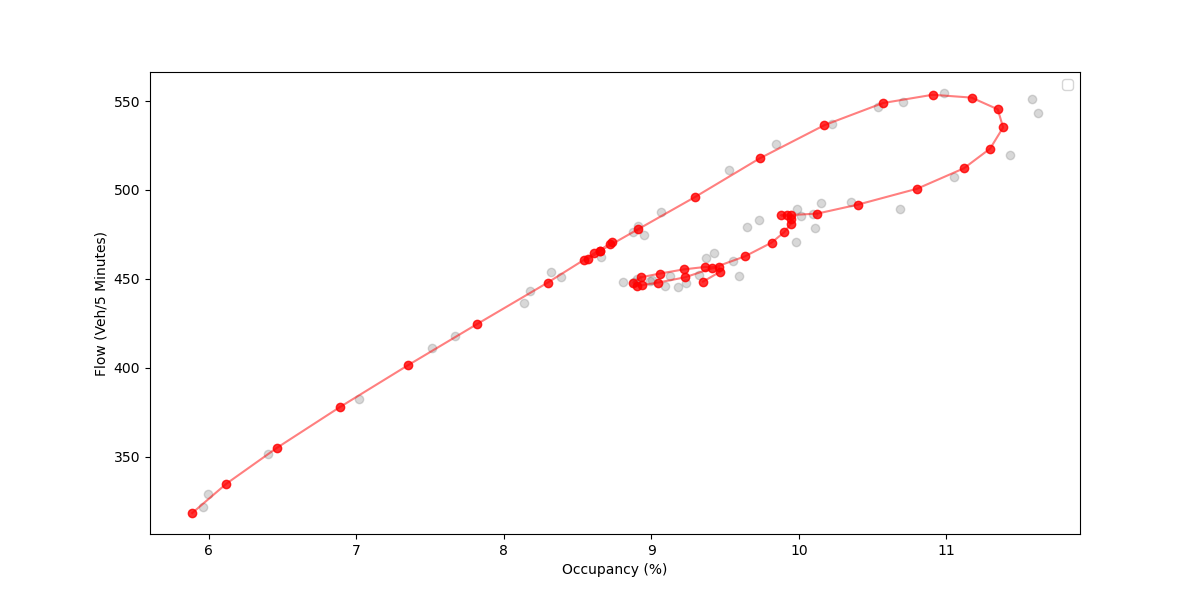} 
    \caption{Smoothed flow-density relationship, averaged from two centrally located detectors.}
    \label{fig:estimate_bn_flow}
\end{figure}

A triangular shape is assumed for the fundamental diagram. The free-flow speed corresponds to the speed limit of this road section at \SI{70}{\mile\per\hour} (approximately \SI{112}{\kilo\meter\per\hour}). The jam density is calculated from the average number of lanes, \(4.7\), multiplied by the inverse of the average vehicle length (we assume a typical value of \SI{6.1}{\meter}), resulting in a rounded value of \SI{760}{\per\kilo\meter}.

The critical density of the road section cannot be derived directly from the collected data. Instead, we calculate it using the ratio between the two slopes of the fundamental diagram branches. We assume this ratio to be (-4), which is a common choice in traffic flow literature. This yields a critical density of \SI{152}{\per\kilo\meter} and a maximum flow of \(112 \cdot 152 = 17024\) vehicles per hour. To accurately capture the resulting traffic behavior and to facilitate the numerical interpretability of the results, we set the corridor length such that it can be traversed in free-flow conditions within \SI{5}{\minute}, which corresponds to approximately \SI{9.34}{\kilo\meter}. The initial vehicle distribution is chosen such that the number of vehicles in each cell corresponds to the boundary flow of the first simulated time step.

As in the previous chapter, we measure the hysteresis effect using the area within the MFD curve. We compare different scenarios: First, a base scenario with a single bottleneck at the lower end, modeled after the conditions at the Washington Avenue Off Ramp. Building on this, we investigate how a more uniform distribution of bottlenecks affects hysteresis by simulating variants with 2, 4, or 8 bottlenecks. All other parameters remain unchanged.

The bottlenecks are placed at equal intervals. The distance between bottlenecks equals the distance between the first bottleneck and the upper end of the section. Mathematically, the \(i\)-th bottleneck in a configuration of \(n\) bottlenecks is positioned at \(9.34 \times \frac{i}{n}\).

The capacities of the bottlenecks decrease linearly in the downstream direction. The reference values are the maximum inflow of 6460 vehicles per hour at the upper end and the capacity of 6240 vehicles per hour at the lowest bottleneck. Thus, the capacity of the \(i\)-th bottleneck (in a configuration of \(n\) bottlenecks) is calculated as:
\[
6460 - \left(\frac{i}{n} \times 220\right).
\]
To examine how sensitive the hysteresis metric is to an increased number of entering vehicles, we simulated additional scenarios. Beyond the standard empirically derived inflow values, we investigated two variants: a moderate increase of 3 \% (high demand) and a significant increase of 6\% (very high demand). The inflow was adjusted accordingly at all time points.
This variation resulted in a total of 24 distinct simulation scenarios. The simulation intervals were adjusted based on demand levels: 4 hours for normal, 5 hours for high, and 6.5 hours for very high demand, ensuring complete congestion dissolution in all cases. The results of this investigation are summarized in Table \ref{table:hysteresis1}. 

\begin{table}[h!]
    \centering
    \caption{Impact of Demand Increases on Hysteresis Area}
    \label{table:hysteresis1}
    \begin{tabularx}{\textwidth}{lXXXX}
        \toprule
        \textbf{Demand} & \textbf{1 Bottleneck} & \textbf{2 Bottlenecks} & \textbf{4 Bottlenecks} & \textbf{8 Bottlenecks} \\
        \midrule
        Normal (1.0)  & 1475.58 & 1421.82 & 1324.07 & 1241.22 \\
        High (1.03)   & 9403.07 & 7906.33 & 6908.70 & 6238.75 \\
        Very High (1.06) & 27917.48 & 21489.42 & 17872.48 & 15445.95 \\
        \bottomrule
    \end{tabularx}
\end{table}

Figure \ref{figure:mfd_dem_sens} illustrates representative MFDs for normal and very high demand, and road geometries with single and eight bottlenecks.

\begin{figure}[H]
    \centering
    \begin{subfigure}[b]{0.45\textwidth}
        \centering
        \includegraphics[width=\textwidth]{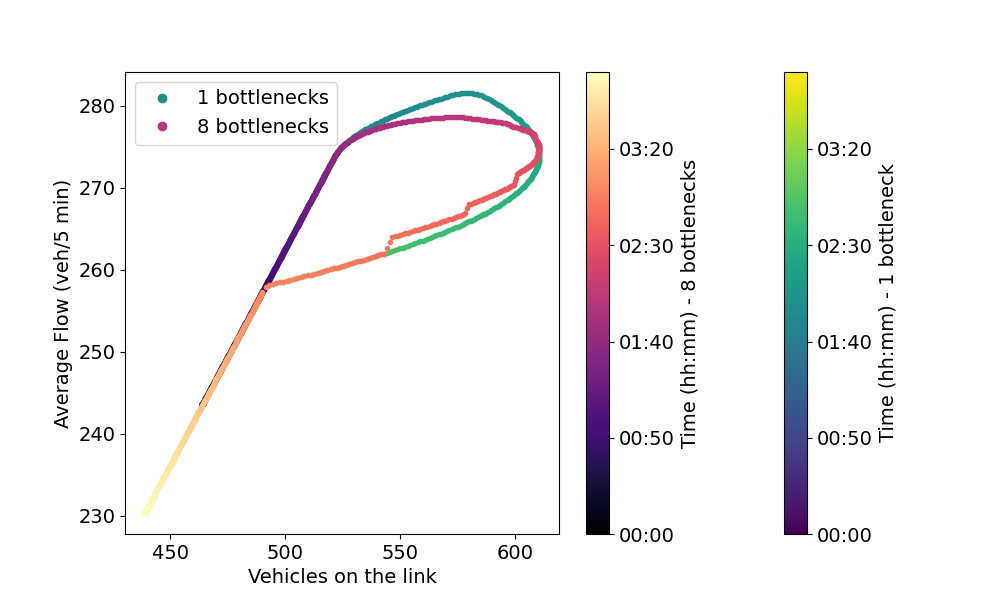}
        \caption{Normal Demand}
        \label{fig:normal_demand}
    \end{subfigure}
    \hfill
    \begin{subfigure}[b]{0.45\textwidth}
        \centering
        \includegraphics[width=\textwidth]{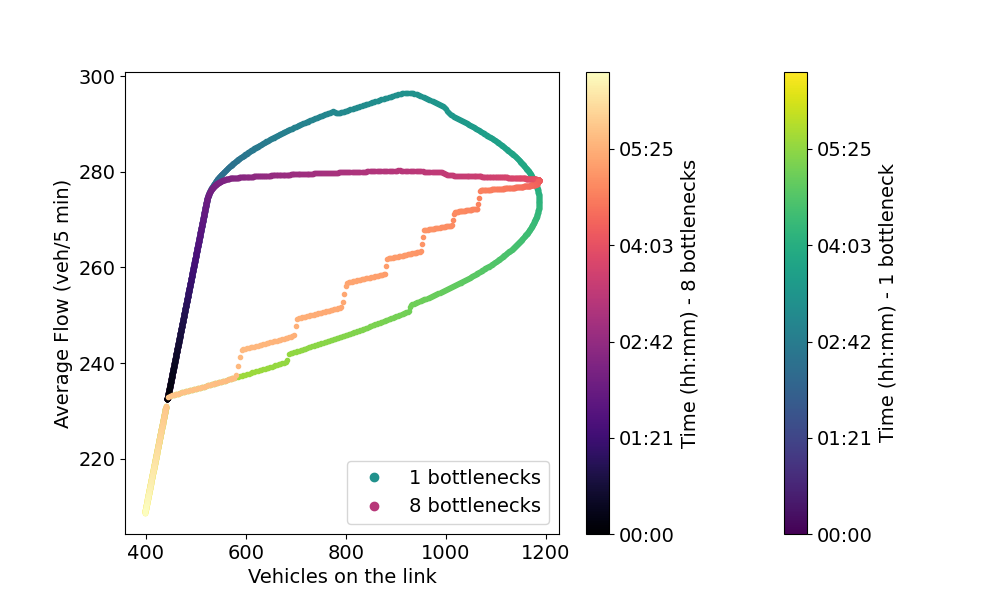}
        \caption{Very High Demand}
    \end{subfigure}
    \caption{MFDs for two demand and bottleneck scenarios.}
    \label{figure:mfd_dem_sens}
\end{figure}

Furthermore, we examine the effects of changing bottleneck capacity by $\pm 3\%$ on the metric under investigation. In the triangular fundamental diagram we studied, hysteresis occurs only during periods of active congestion. This corresponds to the occurrence of capacity drop, which plays a significant role in modeling bottleneck phenomena (cf.~\cite{Hall1991, Chung2007}). The empirically estimated bottleneck capacity of 6240 vehicles corresponds to the value after congestion onset, suggesting that this value already includes a capacity drop $\Delta q_{\text{BN}}$. A capacity drop of 8 \% relative to the initial pre-queuing capacity represents a rough average of values derived from empirical measurements. Using this value as a reference for the initial capacity drop, a subsequent capacity increase of 3\% corresponds to a 34.5\% reduction in the capacity drop, while a capacity reduction of 3\% results in its 34.5\% increase.
We simulated 16 scenarios, again with one to eight bottlenecks. The results are presented in Table~\ref{table:hysteresis2}. MFDs for both cases normal and two representative bottleneck numbers are shown in figure \ref{figure:mfd_cap}.

\begin{table}[h!]
    \centering
    \caption{Impact of changes of bottleneck capacity on hysteresis area}
    \label{table:hysteresis2}
    \begin{tabularx}{\textwidth}{lXXXX}
        \toprule
        \textbf{Bottleneck Capcaity} & \textbf{1 Bottleneck} & \textbf{2 Bottlenecks} & \textbf{4 Bottlenecks} & \textbf{8 Bottlenecks} \\
        \midrule
        6240 veh/hr  & 1475.58 & 1421.82 & 1324.07 & 1241.22 \\
        6427 veh/hr (+3\%)   & 11.74   & 11.74   & 11.74   & 11.74   \\
        6053 veh/hr (-3\%) & 9231.72 & 7791.13 & 6823.08 & 6170.63 \\
        \bottomrule
    \end{tabularx}
\end{table}

\begin{figure}[H]
    \centering
    \begin{subfigure}[b]{0.45\textwidth}
        \centering
        \includegraphics[width=\textwidth]{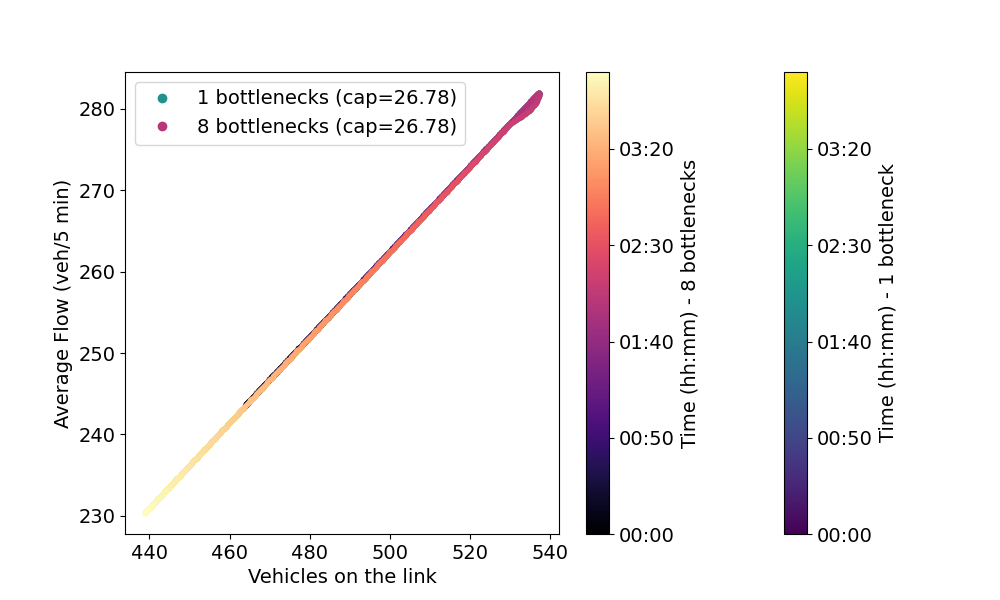}
        \caption{High Capacity Scenario}
    \end{subfigure}
    \hfill
    \begin{subfigure}[b]{0.45\textwidth}
        \centering
        \includegraphics[width=\textwidth]{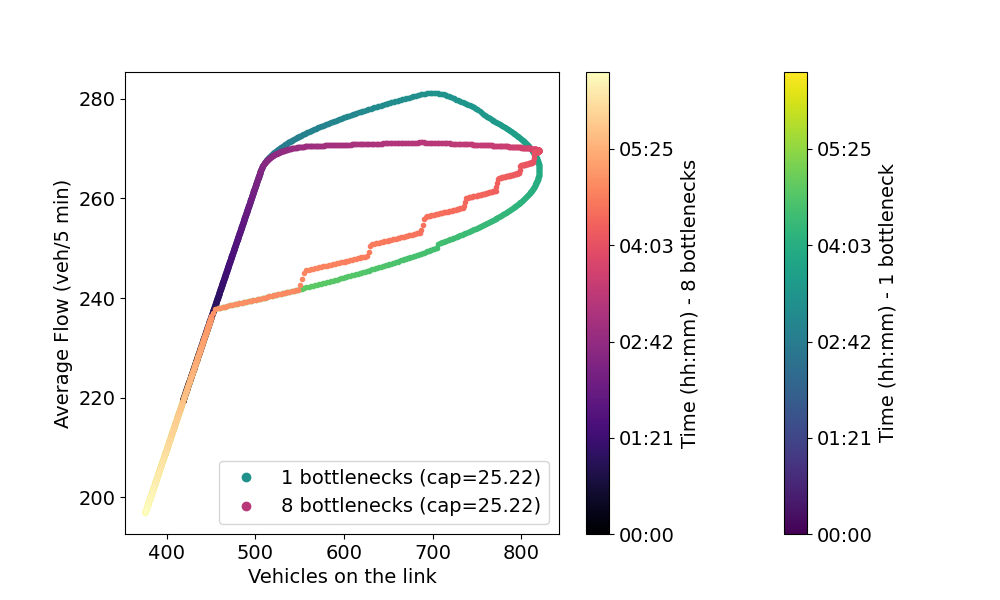}
        \caption{Low Capacity Scenario}
    \end{subfigure}
    \caption{MFDs under variations of bottleneck capacity.}
    \label{figure:mfd_cap}
\end{figure}

The results from both simulation runs demonstrate that the hysteresis area increases drastically even with modest changes in system parameters. For a single bottleneck, a mere three percent increase in demand leads to a 6.37-fold expansion of the hysteresis area, while a further three percent increase results in an additional 2.97-fold expansion. As shown in Proposition \ref{proposition:sensitivity}, there exists a quadratic relationship between demand intensity and hysteresis area. Similarly, a three percent reduction in bottleneck capacity causes a 6.26-fold increase in the area. On the other hand, a three percent increase in capacity leads to an almost complete dissipation of congestion, as shown in the left panel of Figure \ref{figure:mfd_cap}. In this case, additional bottlenecks remain inactive and therefore do not influence the hysteresis metric.

The right panels of Figures \ref{figure:mfd_dem_sens} and \ref{figure:mfd_cap} show distinct discontinuities in the lower part of the hysteresis loop. These occur when a shock wave, triggered by the dissolution of an upstream queue, propagates into and dissolves within the next downstream queue. These discontinuous jumps are absent in the continuous solution approximated by the CTM, representing a numerical artifact that has the opposite effect of the well-known problem of 'numerical diffusion' in the CTM.

The real-world traffic data shows that even during peak hours, the upstream flow varies relatively slowly and only slightly exceeds the bottleneck capacity (6,456 versus 6,240 vehicles per hour). This indicates a delicate balance between road capacity (supply) and traffic demand. Even a minor perturbation in either of these parameters leads to significant system overload. 

As shown in Tables \ref{table:hysteresis1} and \ref{table:hysteresis2}, even minor modifications to the bottleneck geometry can lead to significant hysteresis reduction without compromising overall system performance. Using parameters modeled after the original example, a reduction of 16 percent was achieved, while a hypothetical scenario with stronger congestion showed a reduction of 44.7 percent. The outflow's insensitivity to these adjustments strengthens the robustness of our counterfactual analysis, as it indicates that traffic demand remains stable despite the geometric modifications. Even modest capacity adjustments, achievable through well-established control measures, can thus contribute to improved traffic predictability and generally smoother traffic flow.

\section{Discussion and Conclusions}
\label{sec:conclusion}
This work develops analytical solutions for the time-dependent evolution of spatially aggregated traffic quantities. Empirical analyses, while not providing absolute confirmation, offer robust support for the theoretical predictions within the inherent limitations of real-world traffic data collection and analysis, including environmental noise and potential unobserved confounding variables. Despite these inherent challenges, the strong alignment between theoretical predictions and observed data suggests that our model captures the fundamental dynamics of traffic flow hysteresis at bottlenecks. The assumption of the existence of an MFD significantly simplifies the design of traffic control and management measures, as it allows for the design of time-independent strategies, using only the number of vehicles in the network as the input. The absence of such a diagram is associated with various negative effects, including decreased reliability of travel times and increased traffic instability. The time-dependence of the relationship between macroscopic traffic flow variables complicates the implementation of control measures, as they cannot be adjusted in real time. Through numerical simulations and empirical analysis, we demonstrated that smoother bottleneck geometries reduce traffic flow disruptions and improve control-relevant predictability of the system without harming its performance. As illustrated in Figure \ref{figure:mfd_cap}, a more gradual bottleneck geometry leads to lower gradients in the MFD phase plane and more gradual transitions between macroscopic traffic states. The smoother geometry allows drivers to adjust speeds more gradually, leading to more stable flow patterns during peak hours. 

Besides the geometry of the bottleneck, this work identifies two additional influencing factors on the time-dependence of relationships between macroscopic traffic flow variables:

\begin{enumerate}
    \item \textbf{The shape of the fundamental diagram:} Speed reductions caused by increasing density under uncongested conditions lead to a spatial compression of the average flow towards the upstream end. This effectively reduces the asymmetry of traffic conditions between the on- and offset of congestion. Analogously, this relationship can be explained by the ratio between the traffic flow $q$ and the characteristic speed $\frac{dq}{dk}$, as demonstrated in Proof \ref{proof:sensitivity}: A higher characteristic speed at a given flow causes the flow from the upstream boundary, where the difference between the times of congestion onset and offset at a given accumulation is greatest, to spread faster over the entire uncongested portion of the corridor. The influence of the concavity of the fundamental diagram, as shown in Section \ref{subsec:sensitivity}, is physically comparable to a spatial shift of the bottleneck in the upstream direction and thereby reduces hysteresis for the same reasons as the location-dependent increase in bottleneck capacity examined in Section \ref{sec:simulation2}. Modern applications of kinematic wave theories often assume a triangular fundamental diagram, which implies that traffic flow under uncongested conditions always proceeds with free flow speed, and influences on macroscopic quantities caused by the shape of the fundamental diagram are neglected. During the peak traffic hours analyzed in this text, locations upstream of the bottleneck often reach local flow values close to their capacity $q_c$. The time required for a state with flow $q$ starting at the upstream end of the corridor to spread to its downstream end diverges to $\infty$ as $q$ approaches $q_c$, even if the fundamental relationship $q(k)$ near $q = q_c$ deviates only slightly from a triangular form. In fact, in the case of continuous upstream boundary conditions, a capacity flow $q_c$ originating from the upstream boundary will never reach the downstream end of the corridor if flow moves according to a fundamental diagram that is differentiable at $q = q_c$.

    \item \textbf{Uneven flow increase and decrease rates ($a \neq b$):} In our piece-wise linear model, a necessary condition for the emergence of counter-clockwise hysteresis in the MFD phase plane is as follows: either $a > b$, $q_e > q_b$, or both. A less steep rate of flow decrease induces a more uniform distribution of vehicles, which in turn leads to a higher average flow at the same number of vehicles on the link. \cite{gaydag11} are the only other authors to theoretically explain counter-clockwise hysteresis in the MFD. Their explanation is based on an abstract two-bin model of network traffic, where during the onset, vehicles are loaded into the network at a constant inflow rate, but leave at a rate linearly dependent on the respective flow during the offset. This asymmetry allows, under certain conditions, for a faster decrease in the number of vehicles in the bin with a density above the critical density compared to the bin with lower density, if compared to the loading phase. The model chosen by \cite{gaydag11} can be generalized so that counter-clockwise loops are always possible when the linearly flow-dependent flow increase rate in the loading phase, $p_A$, is less than the flow decrease rate, $p_E$, in the recovery phase. This result is almost diametrically opposite compared to ours, as the model of \cite{gaydag11} does not consider local traffic flow behavior and thus does not incorporate a concept of direction. Moreover, the effects of rapid increase and decrease rates on vehicle distribution and thus average flow in the network are not considered by \cite{gaydag11}. To the best of our knowledge, \cite{He2015} are the only authors who demonstrate the occurrence of figure-eight-shaped hysteresis loops in MFDs empirically. The authors attribute counter-clockwise temporal dynamics, i.e., higher average flows at a fixed number of vehicles in the offset, to a lower rate of lane changes due to less inflow from on-ramps during this phase. As an indication for this explanation, they note that the variance of density values at the onset of congestion is lower than at a time with the same average density in the offset, contrary to the findings of the seminal empirical analysis of MFD hysteresis by \cite{gersun11}. Moreover, \cite{He2015} fit a parabolic curve to point clouds of local flow-occupancy data from two corresponding times and identify a higher local flow-density relationship in the offset phase. It is important to note however that due to the strong skewness of both the spatial density distribution during the offset of congestion and the local flow-density relationship (i.e., $q^{(3)}(k) > 0$), using variance as a measure of the spread of the data might be misleading. During the unloading phase, an active queue was still present at some detector locations $x$, i.e., $k(x,t) \gg \text{avg}_t(k)$ and $q(x) \ge \text{avg}_t(q)$. A similar relationship would arise in our model if the model parameters were chosen such that $a > b$, and in the counter-clockwise part of the MFD diagram, two states at times $t_1$, $t_2$ with $t_2 > t_1$ and equal accumulation are chosen - an active queue would only be observed at $t_2$, and $\bar{q}(t_2) > \bar{q}(t_1)$ applies although the variance of $k(x,t_2)$ is greater than that of $k(x,t_1)$ for $x \in [0,l]$. To prioritize possible causes for the observed effect, a comparison of the variance of density only at detector stations where there is no active queue would be interesting, but is not provided in the discussed research. The authors also neglected the influence of capacity drops and transition zones when comparing the respective corresponding data points, although it is known that these phenomena have a stronger influence in the offset than in the onset of congestion \cite{gersun11}. Therefore, the influence of factors leading to a higher average flow at same density in the offset of congestion should be even stronger than assumed in the article. The suspected relationship between counter-clockwise hysteresis movements and varying frequencies of lane changes is not implausible, but given the lack of empirical validation for this hypothesis, it is important not to dismiss other potential explanations. Specifically, the uniformity of vehicle distribution, particularly during uncongested conditions and when weighted by average flow, remains a plausible and parsimonious explanation based on the provided data.

\end{enumerate}

\appendix

\section{General Solution for the Model}
\label{sec:general_solution}

The following two lemmas comprehensively describe an analytical solution to the described model and, under the given conditions, allow for a simpler computation than the known approaches to the analytical solution of the LWR model. Lemma 2.1 identifies a physically correct solution in uncongested parts of the corridor, while Lemma 2.2 addresses congested conditions.

\begin{lemma}
\label{lem:characteristic} 
	For every point \((x,t)\) satisfying \( x < \psi(t) \) which is reached by at least one characteristic curve, the physically correct characteristic is the latest emanating one.
\end{lemma}
\begin{proof}[Proof]
    The lemma is proven by analyzing a discretized approximation of the upstream boundary condition. Assuming the concavity of \(q(k)\), characteristics may intersect only if originating from the descending part of the boundary. We partition the decreasing branch into intervals \(I_1, \ldots, I_n\). Define \(q_{\text{discr}}(t)\) as \(q_{\text{up}}(t_I)\) where \(t_I\) is the lower bound of the interval \(I\) containing \(t\). Additionally, we linearize \(q(k)\) over the decreasing branch. Suppose two characteristic lines intersect at \((x, t)\), with \(c_1\) from \(I_1\) and \(c_2\) from \(I_2\), \(I_1 < I_2\), and \(c_2\) is most recent. \(c_1\) must have crossed a shockwave, representing the physically valid solution at this point in space-time. The continuous boundary condition solution derives from this discretization method as intervals approach zero length.
\end{proof}

Figure~\ref{fig:approximation_method} illustrates the lemmas approximation method. Figure~\ref{fig:transformation} displays the transformation of a continuously decreasing boundary flow (blue) into three (red) or six (green) discrete steps. Figures~\ref{fig:solution_a} and~\ref{fig:solution_b} show the resulting solutions. These steps propagate as shock waves. The characteristics that intersect the point $(\frac{4}{3}, 25)$ are shown as dotted lines. It is straightforward to verify graphically that only the later emanating characteristic represents a feasible solution of the LWR theory in both cases.

\begin{figure}[H]
	\centering
	\begin{subfigure}{.3\textwidth}
		\centering
		\includegraphics[height=1\linewidth]{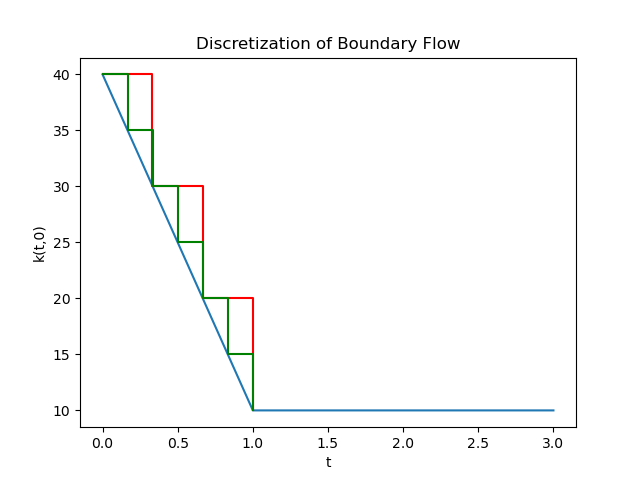} 
		\caption{Discretized flow}
		\label{fig:transformation}
	\end{subfigure}%
	\hfill
	\begin{subfigure}{.3\textwidth}
		\centering
		\includegraphics[height=1\linewidth]{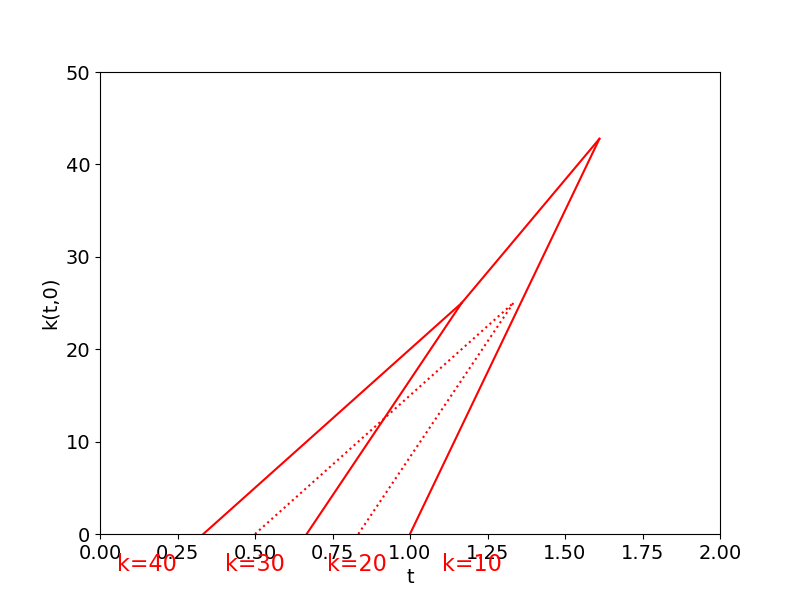} 
		\caption{Approximate solution 1}
		\label{fig:solution_a}
	\end{subfigure}%
	\hfill
	\begin{subfigure}{.3\textwidth}
		\centering
		\includegraphics[height=1\linewidth]{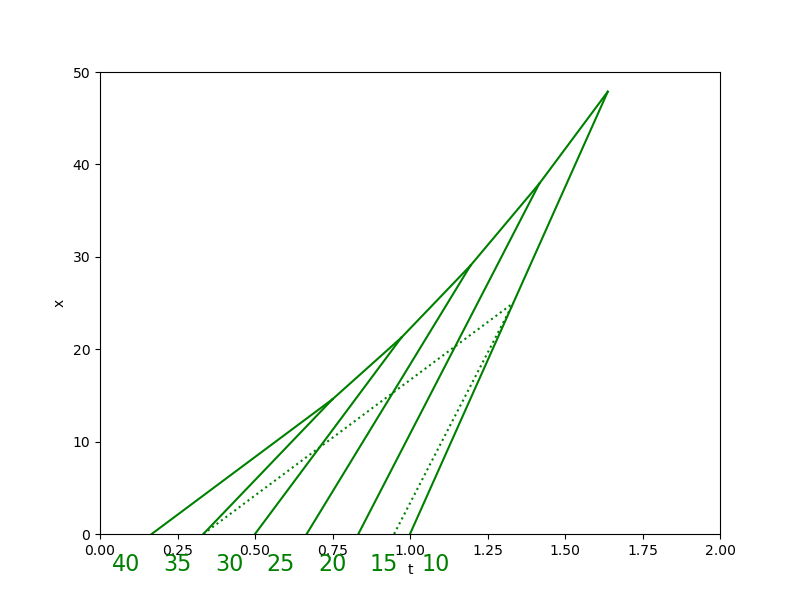} 
		\caption{Approximate solution 2}
		\label{fig:solution_b}
	\end{subfigure}
	\caption{Visualization of the discretized flow and approximate solutions}
	\label{fig:approximation_method}
\end{figure}

\begin{lemma}
    Denote as \( \hat{t} \) the time instant at which the physically correct characteristic reaching $\psi(t)$ at $t$ starts from the upstream boundary.The trajectory of the tail of the queue in space-time, \( \psi(t) \), is then given by the formula:
	\begin{equation}
		\label{eq:psi}
		\psi(t) = l - \left(v(k(0,t_0)) - \frac{dq}{dk}(0,\hat{t})\right) \cdot k(0,t_0(\psi(t), t)) - N_0 + q_b \cdot (t - t_b),
	\end{equation}
\end{lemma}

\begin{proof}[Proof]
	Consider two paths in space-time to reach the point $(\psi(t), t)$: The first path progresses from $(0,0)$ to $(l,0)$, continues to $(l,t)$, and concludes at $(\psi(t), t)$. The second path starts at $(0,0)$, moves directly to $(0,\hat{t})$, and then reaches $(\psi(t), t)$. Since both paths lead from $(0,0)$ to $(\psi(t), t)$, they must both yield the value $N(\psi(t), t)$.
    By setting the integrals along these paths equal and solving for $\psi(t)$, the desired formula is obtained.

\end{proof}

\section{Proof of proposition \ref{proposition:sensitivity} and proposition \ref{proposition:NEF_max}}
\label{section:appendix_proofs}

\begin{proof}[Proof of proposition \ref{proposition:sensitivity}]
\label{proof:sensitivity}
\textbf{Part (a)}: Fix two time points \( t_1 \) and \( t_2 \) with \( t_1 < t_2 \), both having the same accumulation in the interval with an active queue and denote $\bar{q}_f(t)$ the spatial average of local flow on the uncongested portion of the link at time $t$. Then, the following expressions hold:
\[
\bar{q}(t_1) = \left(\frac{l-\psi(t_1)}{l}\right) q_{bn} + \left(\frac{\psi(t_1)}{l}\right) \bar{q}_f(t_1)
\]
and
\[
\bar{q}(t_2) = \left(\frac{l-\psi(t_2)}{l}\right) q_{bn} + \left(\frac{\psi(t_2)}{l}\right) \bar{q}_f(t_2),
\]
where \( \bar{q}_f(t_1) > \bar{q}_f(t_2) \) since at time \( t_2 \) the offset of congestion (\( q(0,t) < q_{bn} \)) has already started. To maximize \( \bar{q}(t_1) - \bar{q}(t_2) \), the queue length at \( t_1 \) and \( t_2 \), \( \psi(t_1) \) and \( \psi(t_2) \), respectively, must be minimized. Therefore, regardless of the specific functional form of the fundamental diagram, the extent of hysteresis increases with the queuing density \( \max q^{-1}(q_{bn}) \). \( \bar{q}(t) = \bar{q}_f(t) \) can then be assumed due to the resulting marginal spatial spread of the queue.

Let \( t_{\text{on}} \) be the time when a queue first forms, and \( t_{\text{off}} \) the earliest time at which the queue has completely dissipated. The extent of hysteresis is given by
\[
\begin{aligned}
H(t) &= \int_{A(t_{on})}^{A(t_{max})} \bar{q}_{\text{on}}(A) \, dA - \int_{A(t_{max})}^{A(t_{off})} \bar{q}_{\text{off}}(A) \, dA \\
&= \int_{t_{\text{on}}}^{t_{\max}} \bar{q}(t) A'(t) \, dt - \int_{t_{\max}}^{t_{\text{off}}} \bar{q}(t) (-A'(t)) \, dt.
\end{aligned}
\]
Assume an arbitrary form for \( q(k) \) such that there exist \( \tilde{t}_1 \) and \( \tilde{t}_2 \) for which \( q'(0, \tilde{t}_1)=v_f \) and \(v_f > q'(0, \tilde{t}_2) \). Denote by \( \Delta \tau \) the difference in travel time of the characteristic lines leaving the upstream boundary at \( \tilde{t}_1 \) and \( \tilde{t}_2 \), respectively. Define
\[
Q^{\text{on}} := \int_{t_{\text{on}}}^{t_{\max}} \bar{q}(t) A'(t) \, \mathrm{d}t,
\]
and
\[
Q^{\text{off}} := \int_{t_{\max}}^{t_{\text{off}}} \bar{q}(t) (-A'(t)) \, \mathrm{d}t.
\]
Note that in intervals with an active queue, it always holds that \( A'(t) = q(0,t) - q_{\text{bn}} \), and therefore both $
\int_{t_{\text{on}}}^{r_1} A'(t) \, \mathrm{d}t$ and $\int_{r_2}^{t_{\text{off}}} A'(t) \, \mathrm{d}t$ are independent of the choice of \( q(k) \) for all \( r_1, r_2 \leq t_{\text{off}} \). To show that the extent of hysteresis increases when \( q(k) \) is adjusted such that \( \frac{\mathrm{d}q}{\mathrm{d}k}(0, \tilde{t}_2)=v_f \), it is therefore sufficent to demonstrate that the change in \( Q^{\text{on}} \), \( \Delta Q^{\text{on}} \), caused by this adjustment is greater than the change in \( Q^{\text{off}} \), \( \Delta Q^{\text{off}} \).

We distinguish between two cases:
\begin{enumerate}
  \item If \( a \geq b \): At the onset of congestion, the average flow at times \( t \) with \( t \geq t_1 + \Delta\tau \) increases by \( 0.5a \) until \( t_{pe} \), within an interval at least as large as \( (q_p-q_0)/a \). In the interval \( [t_{pe}, t_{bn}] \), the flow increases by \( 0.5b \), the length of the interval being \( (q_p-q_0)/b \). In summary, the change in the integral of the average flow at the onset of congestion can be expressed as:
  \[
  \Delta Q^{\text{on}} \geq \frac{1}{2}a\left(\frac{q_p-q_0}{a}\right) - \frac{1}{2}a\Delta\tau + \frac{q_p-q_b}{b} \cdot \frac{1}{2}b = \frac{q_p-q_0}{2} - \frac{1}{2} a \Delta\tau + \frac{q_p-q_b}{2}.
  \]
  The change in accumulated average flow during the offset of congestion is similarly given by
  \[
  \Delta Q^{\text{off}} \leq \frac{q_p-q_0}{2} - \frac{1}{2}b\Delta\tau - \frac{q_p-q_b}{2}.
  \]
  It follows that
  \[
  \Delta Q^{\text{on}} - \Delta Q^{\text{off}} \geq \frac{1}{2}b\Delta\tau - \frac{1}{2}a\Delta\tau + (q_p-q_b).
  \]
  Assuming the congestion behind the bottleneck has already begun at the time when the inflow at the upstream end falls below the bottleneck's capacity, the following maximum travel time associated with the bottleneck flow \( q_{bn} \) holds:
  \[
  \Delta\tau(q_0) \leq \tau(q_0) \leq \tau(q_{bn}) \leq \frac{q_p-q_{bn}}{a}.
  \]
  Inserting into the formula above, we obtain
  \[
  \Delta Q^{\text{on}} - \Delta Q^{\text{off}} \geq (q_p-q_b) - \frac{(a-b)}{a} \cdot (q_p-q_b) \geq 0.
  \]
  Since the integral over \( A'(t) \) in the onset interval must be equal to that in the offset interval, it follows that \( H \) increases due to the described manipulation.
  \item If \( a < b \): Since the flow at the upstream end in this case changes less during the interval \( [t_1, t_1 + \Delta\tau] \) compared to \( [t_2, t_2 + \Delta\tau] \), the following holds:
\begin{align*}
\Delta H &= \Delta \int_{A_0}^{A_{\max}} \bar{q}(A) \, dA - \Delta \int_{A_{\text{off}}}^{A_{\max}} \bar{q}(A) \, dA \\
         &\geq \frac{1}{2} \cdot a \cdot \frac{q_p - q_0}{a} - \frac{1}{2} \cdot b \cdot \frac{q_p - q_0}{b} \\
         &= 0.
\end{align*}

\end{enumerate}

\textbf{Part (b)} \label{proof:prop33}: To derive a more easily interpretable formula, we assume that the travel time \(\tau\) under free-flow conditions is small compared to the duration of the congestion dissipation phase \(t_{pe} - t_{e}\). If this condition is not met, the area under the hysteresis curve will be slightly larger. However, the resulting formula remains quadratic in the peak demand \(q_p\).

We again use the identity
\[
H(t) = \int_{t_{\text{on}}}^{t_{\max}} \bar{q}(t) A'(t) \, dt - \int_{t_{\max}}^{t_{\text{off}}} \bar{q}(t) (-A'(t)) \, dt.
\]

In the asymptotic case of arbitrarily high peak demand, \(H\) can be expressed as

\begin{align*}
H(t) &=  \int_{\tau}^{t_{e}} \bar{q}(t) A'(t) \, dt - \int_{t_{e}}^{t_{\text{off}}} \bar{q}(t) (-A'(t)) \, dt \\
&\in \Theta \left( \int_{\tau}^{t_{e}} \bar{q}(t) A'(t) \, dt \right) = \Theta \left( \int_{\tau}^{t_{e}} \bar{q}(t) (q(0,t) - q_{bn}) \, dt \right) \\
&= \Theta \Bigg( \int_{\tau}^{t_{pb}} \bar{q}(t) (q(0,t) - q_{bn}) \, dt 
+ \int_{t_{pb}}^{t_{pb+\tau}} \bar{q}(t) (q(0,t) - q_{bn}) \, dt \\
& \quad + \int_{t_{pb}+\tau}^{t_{pe}} \bar{q}(t) (q(0,t) - q_p) \, dt 
+ \int_{t_{pe}}^{t_{pe+\tau}} \bar{q}(t) (q(0,t) - q_{bn}) \, dt  \\
& \quad + \int_{t_{pe}+\tau}^{t_{e}} \bar{q}(t) (q(0,t) - q_p) \, dt \Bigg).
\end{align*}

The first equality holds because, in the case of arbitrarily high peak demand, congestion at the bottleneck starts as soon as the vehicle starting at \(t=0\) reaches the downstream end. Additionally, the upstream boundary flow only falls below the capacity of the bottleneck at \(t=t_e\), hence \(t_{\max}=t_e\).

The second equality holds because, for \(t \ge t_e\), neither the average flow nor \(A'(t)\) depend on \(q_p\). Therefore, the term \(\int_{t_{e}}^{t_{\text{off}}} \bar{q}(t) (-A'(t)) \, dt\) has no impact on the complexity class of \(H(t)\).

The third equality holds because, during periods of active congestion downstream of the bottleneck, \(A'(t) = q(0,t) - q(l,t) = q(0,t) - q_{\text{bn}}\).

The fourth equality results from dividing the integral into the individually calculated intervals. 

Since the queue occupies only a marginal spatial portion of the corridor, its effect on the spatial average of the flow can be ignored. Hence, \(\bar{q}(t)\) can be calculated as the average value of the upstream boundary flow between \(t\) and \(t - \tau\).

The individual summands of the integral can then be evaluated as follows. Note that both \(a = \frac{q_p - q_b}{t_{pb}}\) and \(b = \frac{q_p - q_e}{t_e - t_{pe}}\) depend linearly on \(q_p\). Therefore, their coefficients must be considered as part of the quadratic growth of \(H\) when squared or multiplied by \(q_p\). The formula given in the proposition can then be derived by summing the individual components and collecting terms.

\begin{enumerate}[label=\alph*.]
    \item
    \begin{align*}
    &\int_{\tau}^{t_{pb}} \bar{q}(t) (q(0,t) - q_{\text{bn}}) \, dt 
    = \int_{\tau}^{t_{pb}} (q(0,t) - 0.5 a \tau) (q(0,t) - q_{\text{bn}}) \, dt \notag \\
    &= \int_{\tau}^{t_{pb}} (q(0,t) - 0.5 a \tau) q(0,t) \, dt + \Theta (q_p) \notag \\
    &= \int_{\tau}^{t_{pb}} (q(0,t) - 0.5 a \tau) q(0,t) \, dt +\Theta(q_p) \notag \\
    &= \int_{\tau}^{t_{pb}} \left( (q_b + a t)^2 - 0.5 a \tau (q_b + a t) \right) \, dt +\Theta (q_p) \notag \\
    &= \int_{\tau}^{t_{pb}} \left( a^2 t^2 - 0.5 a^2 \tau^2 \right) \, dt +\Theta(q_p) \notag \\
    &=  a^2 \left( \frac{(t_{pb} - \tau)^3}{3} \right) - 0.5 a^2 \tau^2 (t_{pb} - \tau) + \Theta (q_p)
\end{align*}

    \item \begin{align*}
        \int_{t_{pb}}^{t_{pb} + \tau} \bar{q}(t) (q(0,t) - q_{\text{bn}}) \, dt &\in   \tau \frac{q_p^2}{2} + \Theta(q_p)
    \end{align*}

    \item \begin{align*}
        \int_{t_{pb}}^{t_{pb} + \tau} \bar{q}(t) (q(0,t) - q_{\text{bn}}) \, dt &\in   q_p^2 \cdot (t_{pe}-t_{pb}+\tau) +\Theta(q_p)
    \end{align*}

    \item
    \begin{align*}
        &\int_{t_{pe}}^{t_{pe}+\tau} \bar{q}(t) (q(0,t) - q_{\text{bn}}) \, dt \\&\in \int_{t_{pe}}^{t_{pe}+\tau} (q_p - b \cdot \frac{t - t_{pe}}{2} (q_p - b \cdot (t - t_{pe})) \, dt + \Theta(q_p) \notag \\
        &= q_p^2 \tau - \frac{3}{4} q_p b \tau^2 + \frac{b^2 \tau^3}{6}+\Theta(q_p)
    \end{align*}

    \item
    \begin{align*}
        &\int_{t_{pe+\tau}}^{t_{e}} \bar{q}(t) (q(0,t) - q_{\text{bn}}) \, dt \in q_p^2 (t_e - (t_{pe} + \tau)) - q_p b \left( (t_e - t_{pe})^2 - \tau^2 \right)\\& + \frac{b^2}{3} \left( (t_e - t_{pe})^3 - \tau^3 \right) - 0.5 b \tau q_p (t_e - (t_{pe} + \tau))\\& + 0.25 b^2 \tau \left( (t_e - t_{pe})^2 - \tau^2 \right)
        +\Theta(q_p)
    \end{align*}

\end{enumerate}

\end{proof}

\begin{proof}[Proof of proposition \ref{proposition:NEF_max}]
Under the given conditions, the higher of these two values is reached at the previously defined time $t_\text{max}$. For the lower value, we introduce the notation $t_\text{min}$, noting that $t_\text{min} > t_\text{max}$. To identify the appropriate form of the fundamental diagram, we seek:

\begin{equation*}
\begin{aligned}
\arg \max \{&A(t_\text{max}) - A(t_\text{min})\} = \arg \max \{N(0, t_\text{max}) - N(l, t_\text{max}) \\
&\quad - (N(0, t_\text{min}) - N(l, t_\text{min}))\} = \arg \max \{N(0, t_\text{max}) - N(0, t_\text{min}) \\
&\quad - (N(l, t_\text{max}) - N(l, t_\text{min}))\} = \arg \max \{N(0, t_\text{max}) - N(0, t_\text{min}) \\
&\quad - q_\text{bn} \cdot (t_\text{max} - t_\text{min})\} = \arg \max \{q_\text{bn} \cdot t_\text{min} - N(0, t_\text{min})\}
\end{aligned}
\end{equation*}

The last transformation results from the fact that the time $t_\text{max}$ is independent of the form of the fundamental diagram. $t_\text{min}$ is the time at which the queue completely dissipates, thus optimizing the above expression requires maximizing $t_\text{min}$. Since integration along different paths leads to the same result for the cumulative flow, we have $
N(0,t_{\text{min}}) + A(t_{\text{min}}) = N(1,t_{\text{bn,start}}) + q_{\text{bn}} \cdot (t_{\text{min}} - t_{\text{bn,start}}),
$ where \(t_{\text{bn,start}}\) denotes the time of the onset of the bottleneck. This implies that \(t_{\text{min}}\) decreases in \(A(t_{\text{min}})\) and \(N(1,t_{\text{bn,start}})\), while it increases in \(t_{\text{bn,start}}\). The time \(t_{\text{bn,start}}\) can be maximized, without affecting the other mentioned parameters, by maximizing the travel time of the characteristic associated with flow \(q_{\text{bn}}\),
$\frac{dk}{dq}(q_{\text{bn}}),$ under the given conditions. This leads to $\frac{dk}{dq}(q_{\text{bn}}) = t_{\text{pb}} - t_{\text{bn}},$
in the case of maximum NEF hysteresis,
where \(t_{\text{bn}} := \min \{t \mid q(0,t) = q_{\text{bn}}\}\). To determine the form of \(k(q)\) for \(q < q_{\text{bn}}\), we note that \(N(1,t_{\text{bn,start}}) = N(0,t_{\text{bn,start}}) + l \cdot \frac{dN}{dx}\), where \(\frac{dN}{dx}(q_{\text{bn}}) = q_{\text{bn}} \cdot \frac{dk}{dq} - k\). Thus, an increase of \(k(q_{\text{bn}})\) by one unit leads to an increase of \(N(l,t_{\text{bn,start}})\) by \(l\) units. If the travel time of the characteristic leaving the upstream end at \(t_{\text{min}}\) is zero, it is, according to \hyperref[lem:characteristic]{Lemma A.1}, always the physically correct characteristic, and we have \(A(t_{\text{min}}) = 0\). Furthermore, we have \(k(q_{\text{bn}}) = \int_0^{q_{\text{bn}}} \frac{dk}{dq} dq \leq k(q(0,t_{\text{min}})) + (t_{\text{pb}} - t_{\text{bn}}) \cdot (q_{\text{bn}} - q(0,t_{\text{min}}))\), where the inequality is based on the convexity of \(k(q)\). To maximize the NEF hysteresis, it should hold that \(\frac{dk}{dq}(q) = \frac{t_{\text{pb}} - t_{\text{bn}}}{l}\) for \(q > q(0,t_{\text{min}})\), in order to minimize \(\frac{dN}{dx}(q_{\text{bn}})\). An increase in \(k(q)\) in the interval \([0, q(0,t_{\text{min}})]\) by \(r\) results in a reduction of the cumulative flow at the downstream end by \(r \cdot l\) at time $t_{\text{bn,start}}$, but simultaneously leads to an increase in downstream cumulative flow at \(t_{\text{min}}\) by at least the same amount. This follows from the fact that the new arrival time of the characteristic starting at \(t_{\text{min}}\), denoted as \(t_{\text{min}}^+\), yields \(N(l,t_{\text{min}}^+) \geq r \cdot l\) when integrating along this characteristic, and \(A'(t) \leq 0\) holds in this interval (whether the actual cumulative flow \(N(l,t_{\text{min}}^+)\) is lower due to a shockwave is irrelevant for this argument). Thus, the form of the fundamental diagram that maximizes hysteresis is:
\(k(q) = 0\) for \(0 \leq q \leq q(0,t_{\text{min}})\), and
\(k(q) = (t_{\text{pb}} - t_{\text{bn}}) \cdot (q - q(0,t_{\text{min}}))\) for \(q(0,t_{\text{min}}) \leq q\). For the hysteresis metric, we obtain \(A(t_{\text{max}}) - A(t_0) = A(t_{\text{max}}) = N(0,t_{\text{max}}) - N(l,t_{\text{pb}}) - (t_{\text{max}} - t_{\text{pb}}) \cdot q_{\text{bn}}\). A closed-form representation using only the variables defined in previous chapters is unfortunately not possible. 
\end{proof}

\section*{CRediT authorship contribution statement}

Alexander Hammerl; Conceptualization, Methodology, Formal analysis and coding, Data curation, Investigation, Validation, Visualization, Writing; Original draft, review and editing.

Ravi Seshadri; Conceptualization, Supervision, Writing; review

Thomas Kjær Rasmussen; Conceptualization, Review and Funding acquisition.

Otto Anker Nielsen; Conceptualization, Review, Funding acquisition and project management

\section*{Declaration of Competing Interests}
The authors declare that they have no known competing financial interests or personal relationships that could have appeared
to influence the work reported in this paper.

\section*{Declaration of generative AI and AI-assisted technologies in the writing process}
During the preparation of this work the authors used ChatGPT, ClaudeAI in order to improve stylistic aspects and the overall readability of the submitted documents. After using these services, the authors reviewed and edited the content as needed and take full responsibility for the content of the published article.

 \bibliographystyle{elsarticle-num} 
 \bibliography{cas-refs}





\end{document}